\documentclass[11pt]{article}
\usepackage[margin=1in,letterpaper]{geometry}
\usepackage[T1]{fontenc}
\usepackage[utf8]{inputenc}
\usepackage{lmodern}
\usepackage{titling}
\usepackage{natbib}
\setcitestyle{authoryear,round,citesep={;},aysep={,},yysep={;}}
\usepackage{url}
\usepackage{aliascnt}
\usepackage{float}
\usepackage{amssymb}
\usepackage{amsmath}
\usepackage{mathbbol} \DeclareSymbolFontAlphabet{\mathbb}{AMSb}
\DeclareSymbolFontAlphabet{\mathbbm}{bbold}
\usepackage{mathtools}
\usepackage{amsthm}
\usepackage{bm}
\usepackage{enumitem}
\usepackage{xspace}
\usepackage{etoolbox}

\pretocmd{\maketitle}{\hypersetup{pageanchor=false}\begin{titlingpage}}{}{}
\AfterEndEnvironment{abstract}{\end{titlingpage}\hypersetup{pageanchor=true}}
\usepackage{xcolor}
\usepackage{tikz}
\usepackage{algorithm}
\usepackage[noend]{algpseudocode}
\usepackage{needspace}
\usepackage[normalem]{ulem}

\usepackage[hidelinks]{hyperref}
\usepackage[nameinlink,capitalize]{cleveref}
\creflabelformat{equation}{(#2#1#3)}
\crefformat{equation}{(#2#1#3)}
\Crefformat{equation}{(#2#1#3)}
\crefrangeformat{equation}{(#3#1#4) to (#5#2#6)}
\Crefrangeformat{equation}{(#3#1#4) to (#5#2#6)}
\crefmultiformat{equation}{(#2#1#3)}{ and (#2#1#3)}{, (#2#1#3)}{, and (#2#1#3)}
\Crefmultiformat{equation}{(#2#1#3)}{ and (#2#1#3)}{, (#2#1#3)}{, and (#2#1#3)}
\crefname{algorithm}{algorithm}{algorithms}
\Crefname{algorithm}{Algorithm}{Algorithms}

\newtheorem{theorem}{Theorem}[section]

\newaliascnt{lemma}{theorem}
\newtheorem{lemma}[lemma]{Lemma}
\aliascntresetthe{lemma}
\crefname{lemma}{lemma}{lemmas}
\Crefname{lemma}{Lemma}{Lemmas}

\newaliascnt{definition}{theorem}
\newtheorem{definition}[definition]{Definition}
\aliascntresetthe{definition}
\crefname{definition}{definition}{definitions}
\Crefname{definition}{Definition}{Definitions}

\newaliascnt{proposition}{theorem}
\newtheorem{proposition}[proposition]{Proposition}
\aliascntresetthe{proposition}
\crefname{proposition}{proposition}{propositions}
\Crefname{proposition}{Proposition}{Propositions}

\newaliascnt{corollary}{theorem}
\newtheorem{corollary}[corollary]{Corollary}
\aliascntresetthe{corollary}
\crefname{corollary}{corollary}{corollaries}
\Crefname{corollary}{Corollary}{Corollaries}

\newaliascnt{claim}{theorem}

\aliascntresetthe{claim}
\crefname{claim}{claim}{claims}
\Crefname{claim}{Claim}{Claims}
\newtheorem*{claim*}{Claim}

\theoremstyle{definition}
\newaliascnt{construction}{theorem}
\newtheorem{construction}[construction]{Construction}
\aliascntresetthe{construction}
\crefname{construction}{construction}{constructions}
\Crefname{construction}{Construction}{Constructions}

\newaliascnt{remark}{theorem}

\aliascntresetthe{remark}
\crefname{remark}{remark}{remarks}
\Crefname{remark}{Remark}{Remarks}
\newtheorem*{remark*}{Remark}

\theoremstyle{definition}

\newcommand{\defemph}[1]{\emph{#1}}

\renewcommand{\enspace}{\ensuremath{\thinspace}}

\newcommand{\supp}{\operatorname{supp}}
\newcommand{\prune}{\operatorname{prune}}

\newcommand{\one}{\mathbbm{1}}

\newcommand{\Local}{\mathsf{LocalDense}}
\newcommand{\Search}{\mathsf{Search}}

\newcommand{\Peel}{\mathsf{Peel}}

\algrenewcommand{\algorithmicrequire}{\textbf{Input:}}
\algrenewcommand{\algorithmicensure}{\textbf{Output:}}
\algnewcommand{\Parameters}{\item[\textbf{Parameters:}]}
\algrenewcommand{\alglinenumber}[1]{\footnotesize #1:}
\algrenewcommand{\algorithmicindent}{1.2em}
\makeatletter
\def\theHALG@line{\thealgorithm.\arabic{ALG@line}}

\makeatother

\begin{document}

\title{Constant-Factor Local Approximation\\for Bipartite Densest Subgraph}

\author{
Michael Kapralov \\ EPFL
\and Jakub \L\k{a}cki \\ Google
\and Silvio Lattanzi \\ Google
\and Guy Weissenberg \\ EPFL
}
\date{}

\maketitle

\begin{abstract}
Given a vertex $v$ in a large bipartite graph $G$, can one quickly find a small densely connected region near $v$? This is the algorithmic core of many fraud detection, community detection, and recommender systems. One would like a {\em local} algorithm, i.e., an algorithm that probes $G$ in only a small number of locations to find such a region.
Our main contribution is the first constant-factor local approximation, independent of the maximum degree of the graph, for the classical notion of bipartite density: if the sides of a bipartite subgraph are $S$ and $T$, its density is $d(S,T)\coloneq e(S,T)/\sqrt{|S||T|}$, where $e(S,T)$ is the number of edges between $S$ and $T$. We consider \emph{cohesive} dense subgraphs, since every dense subgraph contains a cohesive core of at least the same density. For such a cohesive dense subgraph $S\mathbin{\dot\cup}T$ with at most $k$ vertices per side, our algorithm, starting from any vertex in $S\mathbin{\dot\cup}T$, uses $O_\epsilon(k^{4+\epsilon})$ neighbor queries to find a subgraph $X\mathbin{\dot\cup}Y$ of density $\Omega_\epsilon(d(S,T))$ on $O_\epsilon(k^{2+\epsilon})$ vertices, for every fixed $\epsilon>0$. We improve upon the result of~\cite{Andersen2010}, whose approximation guarantee depends logarithmically on the maximum degree. We complement our positive result with information-theoretic and computational limitations: the quadratic dependence on $k$ in the output size is essentially unavoidable for constant-factor approximation.
Our main technical innovation is an adaptation of the Lov\'asz--Simonovits curve to the bipartite density problem, based on weak-$\ell_p$ norm bounds.
\end{abstract}

\section{Introduction}
\label{sec:introduction}

Identifying dense bipartite subgraphs is a fundamental primitive across machine 
learning, network science, and web data mining. Bipartite graphs model interactions between two kinds of entities: accounts and products 
in e-commerce, users and multimedia content on social platforms, authors and venues in 
bibliometrics, or queries and destination URLs in web search 
\citep{dhillon2001co}. In fraud detection and cybersecurity, researchers search for dense
bipartite subgraphs to detect botnets, astroturfing campaigns, and lockstep review
manipulation, where malicious actors coordinate attacks against targeted subsets 
of resources \citep{BeutelEtAl2013,JiangEtAl2014,HooiEtAl2016,ShinHooiFaloutsos2016}.
In web mining, dense subgraphs correspond to link-spam farms \citep{gibson2005discovering}
and authoritative topic hubs \citep{kleinberg1999authoritative}. 
In computational biology, dense subgraphs are used to identify 
regulatory motifs in genomic DNA \citep{fratkin2006motifcut}, analyze drug--target 
interaction graphs, and uncover multi-functional complexes in gene annotation 
networks \citep{saha2010dense}.

To quantify how tightly two sets of vertices are connected, we use the bipartite density objective \citep{KannanVinay1999,Andersen2010}. Throughout, $G=(V,E)$ is a finite simple undirected bipartite graph with a fixed bipartition $V=L\mathbin{\dot\cup}R$ and maximum degree $\Delta$. For nonempty vertex sets $S\subseteq L$ and $T\subseteq R$ with $e(S,T)$ edges between them, their density is
\begin{equation*}
d(S,T)\coloneq\frac{e(S,T)}{\sqrt{|S||T|}}=\sqrt{\frac{e(S,T)}{|S|}\cdot\frac{e(S,T)}{|T|}},
\end{equation*}
the geometric mean of the average degree from $S$ into $T$ and the average degree from $T$ into $S$.

The maximum of $d(S,T)$ over all nonempty sets $S,T$ on opposite sides can be computed in polynomial time by linear
programming \citep{Charikar2000}. However, on modern web-scale networks---frequently
spanning hundreds of millions of vertices and billions of edges \citep{QinEtAl2015}---such
global algorithms are computationally prohibitive, as their running time and memory 
footprint scale with the size of the entire graph. We therefore study the local formulation
\citep{Andersen2010}: given a starting vertex $v$ and a target size parameter $k$,
find a dense bipartite subgraph in the vicinity of $v$ using a number of queries and running time that depend solely on $k$, completely independent of the graph size $|V|$ and maximum degree $\Delta$.

Classical local clustering algorithms for low-conductance sets
\citep{SpielmanTeng2008,AndersenChungLang2006,AllenZhuLattanziMirrokni2013} seek sets with few boundary edges relative to their total degree. In contrast, $d(S,T)$ is independent of edges leaving $S\mathbin{\dot\cup}T$. Dense bipartite subgraphs in interaction graphs often have many edges to the rest of the graph---for example, when fraudulent accounts add large numbers of legitimate interactions as ``camouflage'' \citep{HooiEtAl2016}. With these camouflage edges, the cut conductance of $S\mathbin{\dot\cup}T$ can approach $1$. Standard random-walk and diffusion methods can then quickly leave the target, so we need local algorithms designed specifically for the bipartite density objective (Figure~\ref{fig:two-structures}).

\begin{figure}[t]
\centering
\begin{tikzpicture}[
  scale=0.92,transform shape,
  x=1cm,y=1cm,
  lv/.style={circle,draw=blue!60!black,fill=blue!60!black,inner sep=0pt,minimum size=5.5pt},
  rv/.style={circle,draw=orange!85!black,fill=orange!85!black,inner sep=0pt,minimum size=5.5pt},
  ov/.style={circle,draw=gray!70,fill=gray!70,inner sep=0pt,minimum size=3pt},
  ie/.style={draw=black!85,line width=0.9pt},
  oe/.style={draw=gray!60,line width=0.4pt},
  lab/.style={font=\small}
]
  \foreach \i in {1,2,3} {
    \node[lv] (s\i) at (0,{1.0*(\i-2)}) {};
    \node[rv] (t\i) at (2.0,{1.0*(\i-2)}) {};
  }
  \foreach \i in {1,2,3} {
    \foreach \j in {1,2,4,5} {
      \node[ov] (ls\i\j) at (-1.7,{1.0*(\i-2)+0.2*(\j-3)}) {};
      \draw[oe] (s\i) -- (ls\i\j);
      \node[ov] (rt\i\j) at (3.7,{1.0*(\i-2)+0.2*(\j-3)}) {};
      \draw[oe] (t\i) -- (rt\i\j);
    }
    \foreach \offset in {-0.07,0,0.07} {
      \fill[gray!70] (-1.7,{1.0*(\i-2)+\offset}) circle (0.4pt);
      \fill[gray!70] (3.7,{1.0*(\i-2)+\offset}) circle (0.4pt);
    }
  }
  \foreach \i in {1,2,3} \foreach \j in {1,2,3} \draw[ie] (s\i) -- (t\j);
  \node[lab] at (0,1.55) {$S$};
  \node[lab] at (2.0,1.55) {$T$};
  \node[lab,gray!80!black] at (-1.7,1.55) {\footnotesize $D$ leaves each};
  \node[lab,gray!80!black] at (3.7,1.55) {\footnotesize $D$ leaves each};
\end{tikzpicture}
\caption{A complete bipartite subgraph $(S,T)$ with $k$ vertices in $S$ (blue) and $k$ in $T$ (orange), each with $D\ge1$ leaves (gray) outside $S\mathbin{\dot\cup}T$. Here $d(S,T)=k$ for every $D$. Each vertex of $S\mathbin{\dot\cup}T$ has $D$ of its $k+D$ edges leaving $S\mathbin{\dot\cup}T$, so the conductance of $S\mathbin{\dot\cup}T$ is at least $D/(k+D)$, which tends to one as $D/k\to\infty$. Every star in this graph has density at most $\sqrt{k+D}$, which is $o(k)$ for $D=o(k^2)$.}
\label{fig:two-structures}
\end{figure}
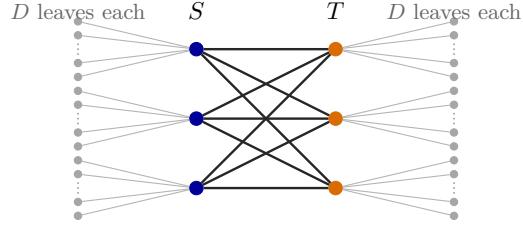

More formally, the local algorithm receives a starting vertex $v$, an integer $k\ge1$, and adjacency-list access to $G$, as in \citet{Andersen2010}. We count neighbor queries: we read a specified entry of an adjacency list, or receive $\bot$ past its end. 
The guarantee is measured against unknown cohesive sets $S\subseteq L$, $T\subseteq R$ (\Cref{def:cohesive}) with at most $k$ vertices on each side. We require the algorithm to succeed from every starting vertex in $S\mathbin{\dot\cup}T$, using a number of queries independent of $|V|$.
An approximation ratio $\alpha\ge1$ means returning nonempty sets $X\subseteq L$, $Y\subseteq R$ with $d(X,Y)\ge d(S,T)/\alpha$. For a single vertex, we write $e(u,T)\coloneq e(\{u\},T)$ and $e(S,w)\coloneq e(S,\{w\})$.

We ask for approximation, not recovery: the union of the returned sets need not contain the starting vertex or intersect $S\mathbin{\dot\cup}T$ (after adding $v$ to the returned set on its side, the density is at least $d(X,Y)/\sqrt2$; see \Cref{lem:distance}). A star with at least $d(S,T)^2$ leaves has density at least $d(S,T)$
and is therefore an acceptable output.

Andersen's algorithm can be implemented using $O(\Delta k^2)$ neighbor queries and returns sets $X,Y$ with $O(\Delta k^2)$ vertices each (this follows from Andersen's analysis \citep[proof of Theorem~2]{Andersen2010}; see also \Cref{app:andersen-queries}). Its approximation ratio is bounded by $O(1+\log_2(\Delta/d(S,T)))$ \citep[Theorem~1]{Andersen2010}. This bound is constant only when $d(S,T)=\Omega(\Delta)$. For the complete bipartite subgraph of \Cref{fig:two-structures} with $D$ camouflage edges per vertex, the lower bound on the returned density is $\Omega(k/(1+\log_2((k+D)/k)))$, which decreases as $D$ grows even though $d(S,T)=k$. The density is invariant under the addition of edges outside the subgraph induced by $S\mathbin{\dot\cup}T$, and the guarantee is not. Whether this dependence on the maximum degree is necessary for local approximation was open.

\subsection{Our results}
\label{sec:results}

Our main result is a constant-factor local approximation algorithm for the bipartite densest subgraph problem. We cannot guarantee success from every vertex of $S\mathbin{\dot\cup}T$ for arbitrary target sets $S,T$, since a vertex in $S\mathbin{\dot\cup}T$ may have few or no edges into the opposite target set. We therefore require every target vertex to have at least half the average degree of its side within the target subgraph.

\begin{definition}[Cohesive sets]
\label{def:cohesive}
Nonempty sets $S\subseteq L$, $T\subseteq R$ with $e(S,T)>0$ are \defemph{cohesive} if $e(u,T)\ge\frac{e(S,T)}{2|S|}$ for all $u\in S$, and $e(S,w)\ge\frac{e(S,T)}{2|T|}$ for all $w\in T$.
\end{definition}

We use the term \emph{cohesiveness} to formalize the minimum-degree structure already implicit in Andersen's proof of his main theorem \citep[Theorem~1]{Andersen2010}. For any sets $S\subseteq L$, $T\subseteq R$ with $e(S,T)>0$, there are cohesive subsets $S^*\subseteq S$, $T^*\subseteq T$ of at least the same density, $d(S^*,T^*)\ge d(S,T)$ (\Cref{lem:cohesive}). For arbitrary $S,T$, our setting is exactly that of Andersen, whose algorithm succeeds from every vertex of subsets that are selected by a minimum-degree argument and contain at least half of the edges (see the discussion in \Cref{cor:arbitrary}).

\Needspace{10\baselineskip}
\begin{theorem}[Main result]
\label{thm:overview}
For every $\epsilon>0$, all cohesive sets $S\subseteq L$, $T\subseteq R$ with $|S|,|T|\le k$, and every starting vertex $v\in S\mathbin{\dot\cup}T$, the deterministic algorithm $\Local_{2+\epsilon}(G,v,k)$ (\Cref{alg:local}) returns nonempty sets $X\subseteq L$, $Y\subseteq R$ with the following guarantees:
\begin{enumerate}[label=(\roman*),nolistsep]
\item \emph{Locality and complexity:} The output $(X,Y)$ has $O_\epsilon(k^{2+\epsilon})$ vertices, and the algorithm makes $O_\epsilon(k^{2+\epsilon}\min\{\Delta\log_2^2(2k),k^2\})$ neighbor queries and runs in time $\widetilde O_\epsilon(k^{2+\epsilon}\min\{\Delta,k^2\})$.
\item \emph{Constant-factor approximation:} $d(X,Y)\ge\frac{\sqrt3}{16\sqrt{\mathrm e}}\frac{\epsilon}{1+\epsilon}\,d(S,T)$, independent of edges outside the subgraph induced by $S\mathbin{\dot\cup}T$.
\end{enumerate}
\end{theorem}

We complement our positive results by proving limitations on the output size:

\begin{theorem}[Limitations on the output size, informal]
\label{thm:limits-informal}
No possibly randomized algorithm with neighbor-query access that, for all cohesive sets $S\subseteq L$, $T\subseteq R$ with $k$ vertices on each side, returns nonempty sets $X\subseteq L$, $Y\subseteq R$ with $d(X,Y)=\Omega(d(S,T))$ with probability at least $2/3$ from at least one vertex in $S\mathbin{\dot\cup}T$ can guarantee an output bound of $o(k^2)$ vertices on each side if either:
\begin{enumerate}[label=(\roman*),nolistsep]
\item \emph{Unconditional:} Its number of queries depends only on $k$; or
\item \emph{Computational:} Assuming the Strongish Planted Clique Hypothesis (SPCH), its running time is at most $f(k)\,|V|^c$ for a constant $c$ and a computable function $f$ of $k$ alone.
\end{enumerate}
\end{theorem}

Up to the factor $k^\epsilon$, our output size is therefore the best possible in both settings. We discuss the main result in \Cref{sec:analysis} and the limitations in \Cref{sec:limits}.

\subsection{Our technique}
\label{sec:techniques}

\paragraph{From power iteration to prefix sweeps.}
Our starting point is Andersen's pruned power iteration \citep{Andersen2010}. We maintain a nonnegative vector $x$ indexed by vertices and multiply it by the adjacency matrix, setting $y(w)\coloneq\sum_u x(u)$, where the sum is over the neighbors of $w$. For now, we set aside rescaling and pruning. The mass of $x$ on a set $Q$ is $\sum_{u\in Q}x(u)$. Suppose every vertex of $S$ has at least $a$ neighbors in $T$ and every vertex of $T$ at least $b$ in $S$. After one multiplication, the mass of $y$ on $T$ is at least $a$ times the mass of $x$ on $S$, and after two, the mass on $S$ is at least $ab$ times its original value. These bounds hold even with edges leaving the target, since all coordinates are nonnegative. By cohesiveness, there exist powers of two $a,b\in[1,k]$, with $\sqrt{ab}>d(S,T)/4$ (see \Cref{sec:vertex-peeling}; we remove this choice later).
We cannot locate a dense pair from the growth of target mass alone: the target is unknown, and mass outside it may grow faster. Andersen extracts a dense pair from growth of the Euclidean norm. Instead we sort consecutive vectors and consider their initial segments, or \emph{prefixes}. We seek a Lov\'asz--Simonovits-type \emph{win--win} \citep{SpielmanTeng2008,AndersenChungLang2006}: either we find a dense prefix pair in a sweep, or we obtain upper bounds on prefix sums that eventually become incompatible with the mass that remains on the target. Here set size means number of vertices, rather than volume.

\paragraph{The choice of norm.}
Which bound on prefix sums should we use for this alternative? Write $P_i(x)$ for the vertices with the $i$ largest positive coordinates, breaking ties by identifier, and $I_x(m)\coloneq\max_{|Q|\le m}\sum_{u\in Q}x(u)$. A natural choice is the Euclidean norm, since $I_x(m)\le\|x\|_2\sqrt m$. However, with only this inequality we incur a logarithmic loss when summing comparable bounds over dyadic ranges of prefix sizes (see \Cref{sec:prefix-sweeps}). Changing the sweep alone is therefore not enough. We instead fix $p>2$ and use the \emph{weak-$\ell_p$ norm} \citep[Chapter~4]{BennettSharpley1988}
\begin{equation*}
\|x\|_{p,\infty}\coloneq
\sup_{m\in\{1,2,\ldots\}}\frac{I_x(m)}{m^{1-1/p}}
\enspace.
\end{equation*}
The bounds over dyadic ranges now form a geometric series, so we incur only a $p$-dependent factor when summing them (a detailed explanation of this comparison can be found in \Cref{app:log-factor}). More precisely, let $y=xA$, where $x\ge0$ is supported on $L$ and $A$ is an $L\times R$ adjacency matrix with column sums at most $B_\infty$. Let $\sigma>0$ be the density we aim to guarantee in the returned sets. For a prefix $P_m(y)$, if $d(P_i(x),P_m(y))\le\sigma$ for every prefix $P_i(x)$, then
\begin{equation*}
I_y(m)\le C_p\,\sigma^{2/p}B_\infty^{1-2/p}
\|x\|_{p,\infty}m^{1-1/p},
\qquad C_p\coloneq\frac{p-1}{p-2}
\enspace.
\end{equation*}
If the inequality fails for $P_m(y)$, we find dense sets by sweeping the prefixes of $x$ against $P_m(y)$. If no prefixes exceed density $\sigma$, we take the supremum over $m$ to bound $\|y\|_{p,\infty}$ (\Cref{sec:prefix-sweeps,app:interpolation}).

\paragraph{Removing the degree dependence.}
The logarithm is gone, but the column sums remain. For the full adjacency matrix, we only know the bound $B_\infty=\Delta$ on the column sums, where $\Delta$ is the maximum degree. Since $\sqrt{ab}>d(S,T)/4$, we obtain a constant-factor approximation if we can guarantee density $\sigma$ equal to a fixed fraction of $\sqrt{ab}$, depending only on $p$. The upper bound on the norm after two steps has factor $C_p^2\sigma^{4/p}\Delta^{2-4/p}$, compared with the factor $ab$ for the mass on $S$. To make the norm factor smaller than $ab$, we need $\sigma<C_p^{-p/2}\sqrt{ab}(\sqrt{ab}/\Delta)^{(p-2)/2}$, at the cost of an additional approximation factor $(\Delta/\sqrt{ab})^{(p-2)/2}$.

After rescaling, we want the target mass to grow while the weak norm decreases over pairs of steps. For a step from $L$ to $R$, we use \emph{peeling} to bound the number of selected edges incident to each vertex in $R$ while keeping enough target edges. We repeatedly remove a vertex in $R$ with at most $b/2$ remaining edges and select those edges, or a vertex in the current support with at most $a/4$ remaining edges and discard those edges. If vertices remain but neither removal is possible, their density exceeds $\sqrt{ab/8}$ (\Cref{sec:vertex-peeling}).
If every vertex is removed, each vertex in $R$ has at most $b/2$ selected edges, while at least $3a/4$ edges into $T$ are selected at each vertex of $S\cap\supp(x)$. Let $A$ be the adjacency matrix of the selected edges and set $y=(2/b)xA$, so that the weighted column sums are at most one. For $\sigma\coloneq\sqrt{3ab}/(4C_p^{p/2})$, we either find prefixes of density greater than $\sigma$, or
\begin{equation*}
\sum_{w\in T}y(w)\ge\frac{3a}{2b}\sum_{u\in S}x(u),
\qquad
\|y\|_{p,\infty}\le
\left(\frac{3a}{4b}\right)^{1/p}\|x\|_{p,\infty}
\enspace.
\end{equation*}
For the reverse step, we exchange $a$ and $b$. If no dense sets are returned in either step, then after two steps without pruning the mass on $S$ is at least $9/4$ times its earlier value, while the weak norm is at most $(3/4)^{2/p}$ times its earlier value. Crucially, neither factor now depends on $\Delta$ (\Cref{sec:algorithm}).

\paragraph{The local search.} We are now ready to combine the ideas above and present the algorithm. So far, we have assumed that suitable degree bounds $a,b$ are known. In reality, since they are unknown, we run a search for every pair of powers of two $a,b\in[1,k]$ and return the densest candidate. Assume $v\in L$, exchanging the sides if necessary. Each search starts with $x_0=\one_v$, the vector equal to one at $v$ and zero elsewhere, and alternates between the sides. At step $t$, the algorithm reads at most $\lceil\sigma^2\rceil$ entries of the adjacency list of each vertex in $\supp(x_t)$. If it finds that many neighbors in one list, it returns a star of density at least $\sigma$. If no star is returned, all these adjacency lists have been read completely. The algorithm then peels the edges incident to the support of $x_t$ and, unless dense sets are returned, performs the rescaled multiplication and sets coordinates below $\delta$ to zero. We choose $\delta$ to keep the support small without losing too much target mass.

We abandon the guess if the pruned vector $z$ is zero and otherwise set $x_{t+1}=z$ only if its weak norm satisfies the bound from the inequalities above, which decreases over pairs of steps. If this bound is exceeded, we sweep $x_t$ against the shortest prefix of $z$ whose sum exceeds the corresponding bound. This prefix's vertices and sum are unchanged by pruning, so we find dense sets.

For the degree bounds chosen using cohesiveness and $v\in S$, we measure target mass with a factor $b/a$ on $T$ and no factor on $S$. Before pruning, this mass is at least $3/2$ times its previous value, and we choose $\delta$ so that the loss during pruning is at most $1/2$. The weighted mass therefore stays at least one, its initial value, and the vector never becomes zero. At every even step reached by the search,
\begin{equation*}
1\le\sum_{u\in S}x_t(u)
\le\|x_t\|_{p,\infty}k^{1-1/p}
\le(3/4)^{t/p}k^{1-1/p}
\enspace.
\end{equation*}
After $O_p(\log_2(2k))$ steps, the right-hand side would be less than one, so the search must return sets of density at least $\sigma$ (\Cref{sec:algorithm}).

\subsection{Related Work}

\citet{SpielmanTeng2008} used the Lov\'asz--Simonovits curve \citep{LovaszSimonovits1990,LovaszSimonovits1993} to analyze their local partitioning algorithm \textsc{Nibble}, and \citet{AndersenChungLang2006} adapted this analysis to personalized PageRank. Other local graph-partitioning algorithms use heat-kernel PageRank \citep{Chung2009} or evolving sets \citep{AndersenPeres2009}, and local clustering by $p$-norm flow diffusions and $p$-norm cuts \citep{FountoulakisWangYang2020,LiuGleich2020} has been extended to vertex attributes and noisy labels \citep{YangFountoulakis2023,BackDeLucaFountoulakisYang2024}. Local algorithms for dense bipartite-like subgraphs \citep{Peng2012,LiPeng2013,MacgregorSun2021} minimize the bipartiteness ratio, a cut objective analogous to conductance.

Local and seeded variants of densest subgraph have also been studied from a practical perspective. When densest subgraphs are extracted sequentially, fragmented sub-regions of a single dominant global core may be returned while distinct local communities are missed. \citet{QinEtAl2015} therefore introduced \emph{locally densest subgraphs} (LDS), based on $\rho$-compact subgraphs, with structural core bounds and pruning techniques for exact verification on billion-edge graphs, but their algorithm traverses the whole graph. \citet{DaiQiaoChang2022} introduced the \emph{anchored densest subgraph} (ADS) problem, with explicit edge or vertex penalties for expanding beyond a specified reference set, and solved it via localized parametric flow networks. \citet{HuangGleichVeldt2024} generalized this formulation to hypergraphs using negative supermodular optimization, with strongly localized routines when the penalization parameter is sufficiently large. For dense-block discovery in fraud detection, \citet{HooiEtAl2016} gave a camouflage-resistant algorithm for the whole graph.

The local problem of \citet{Andersen2010} was revisited by \citet[Theorems~1.2.2 and~1.2.3]{Liu2026}, who obtains an output size of $O(k^2)$ with Andersen's approximation ratio, and $O(k^{1+\epsilon})$ with a ratio that also depends on $|V|$. The techniques involved are different from ours and are based on spectral analysis, degree sorting, and majorization. Our algorithm, with $\epsilon=1/\lceil\log_2(2k)\rceil$, also returns $O(k^2)$ vertices, using $O(k^4)$ neighbor queries with an $O(\log_2(2k))$ approximation ratio, independent of $\Delta$ (\Cref{cor:exponent}).

\section{The local algorithm and its analysis}
\label{sec:analysis}

Our algorithm is based on an alternating power iteration from the starting vertex: with $x_0\coloneq\one_v$, it repeatedly multiplies the current vector by the adjacency matrix of a subgraph of $G$ and then prunes, alternating between the two sides. For suitable degree guesses, the mass inside a cohesive target grows by a factor of at least $9/4$ over two weighted steps, before accounting for pruning. The target is unknown to the algorithm, and the mass outside it may grow as well, so we cannot identify dense sets from growth alone. For the analysis, we compare prefixes: either prefixes $P,Q$ of consecutive vectors have density $d(P,Q)>\sigma$ and the algorithm returns them, or the next prefix sums are bounded in terms of the current ones, with geometrically decreasing bounds over pairs of steps. This cannot last forever: the target has at most $k$ vertices per side, but its mass on the starting side stays at least one at every even step.

We use neighbor queries with fixed adjacency-list orders and distinct vertex identifiers. Vectors are row vectors indexed by vertices, with indicators $\one_S$, $\one_v\coloneq\one_{\{v\}}$ and support $\supp(x)\coloneq\{u:x(u)\ne0\}$. Vectors are restricted or extended by zero to match matrix indices. For any subgraph $H$ of $G$, numbers of edges and densities in $H$ are denoted by $e_H(S,T)$ and $d_H(S,T)$, and $N_G(U)$ denotes the set of neighbors of $U\subseteq V$ in $G$.

\subsection{Prefix sweeps}
\label{sec:prefix-sweeps}

In a step of the iteration from $L$ to $R$, the current vector $x\ge0$ is supported on $L$ and the next vector is $y\coloneq xA$, where $A$ is the adjacency matrix of a subgraph of $G$, so that each coordinate $y(w)$ is the sum of $x(u)$ over the neighbors $u$ of $w$ in this subgraph. To compare the largest coordinates of $y$ with those of $x$, we adapt the Lov\'asz--Simonovits curve \citep{LovaszSimonovits1990,LovaszSimonovits1993} to density: in place of the curve we use $I_x(m)$, the largest mass of the current vector on $m$ vertices, and in place of a sweep set of low conductance we look for prefixes $P,Q$ with $d(P,Q)>\sigma$, where $\sigma>0$ is a density threshold, fixed in \Cref{eq:parameters}.

\begin{definition}[Prefix sums and the weak norm {[\citealp[equation~(3)]{SpielmanTeng2008}; \citealp[Chapter~4]{BennettSharpley1988}]}]
\label{def:prefix}
For a vector $x\ge0$ indexed by a finite set and $\ell\coloneq|\supp(x)|$, write $x_{(1)},\ldots,x_{(\ell)}$ for its positive coordinates in nonincreasing order, with ties broken by identifier, and set $x_{(i)}\coloneq0$ for $i>\ell$. For $1\le i\le\ell$, the prefix $P_i(x)$ consists of the first $i$ vertices in this order. For $p\ge2$, the prefix sums and weak-$\ell_p$ norm are defined by
\begin{equation}
\label{eq:ov-weak}
I_x(m)\coloneq\max_{|P|\le m}\sum_{u\in P}x(u)=\sum_{i=1}^m x_{(i)},\qquad
\|x\|_{p,\infty}\coloneq\sup_{s>0}\frac{I_x(s)}{s^{1-1/p}}
\enspace,
\end{equation}
where $m$ is a nonnegative integer and $I_x$ is extended to real arguments $s\ge0$ by linear interpolation between consecutive integers, with $I_x(0)=0$. In particular, $I_x(s)\le\|x\|_{p,\infty}s^{1-1/p}$ for every $s>0$.
\end{definition}

\begin{definition}[Column-sum bound]
\label{def:column-sum}
A nonnegative matrix $A$ has \defemph{column-sum bound} $B_\infty$ if $\sum_uA_{uw}\le B_\infty$ for every column $w$.
\end{definition}

After sorting $y$, we take a prefix $P_m(y)$, whose sum is $I_y(m)$, and in a \defemph{sweep} we look for a prefix $P_i(x)$ with $d(P_i(x),P_m(y))>\sigma$. If there is none, then there are at most $\sigma\sqrt{im}$ edges between $P_i(x)$ and $P_m(y)$ for every $i$. For $p>2$ and a column-sum bound $B_\infty$ for $A$ (\Cref{def:column-sum}),
\begin{equation}
\label{eq:local-summation}
I_y(m)\le C_p\,\sigma^{2/p}B_\infty^{1-2/p}\|x\|_{p,\infty}m^{1-1/p},
\qquad\text{where }C_p\coloneq\frac{p-1}{p-2}
\enspace.
\end{equation}
This bound is a special case of \Cref{prop:interpolation}, proved in \Cref{app:interpolation} by summation by parts. In the proof, we use the smaller of two bounds for the number of edges from $P_i(x)$ into $P_m(y)$ in the subgraph represented by $A$: $\sigma\sqrt{im}$ from the density assumption and $B_\infty m$ from the column sums. Whenever $I_y(m)$ exceeds the right side of \Cref{eq:local-summation}, the algorithm finds a prefix $P_i(x)$ with $d(P_i(x),P_m(y))>\sigma$ in the sweep. If no such prefixes exist for any $m$, then $\|y\|_{p,\infty}\le C_p\sigma^{2/p}B_\infty^{1-2/p}\|x\|_{p,\infty}$.

The bound holds only for $p>2$: for $p=2$, we obtain a logarithmic factor in place of $C_p$, as in the guarantee of \citet{Andersen2010}, to which we compare ours in \Cref{app:andersen-prefix}. For the adjacency matrix of $G$ itself, whose column sums are the degrees, we can only take $B_\infty\coloneq\Delta$, so that the bound contains a factor $\Delta^{1-2/p}$, even though $d(S,T)$ is independent of edges outside the subgraph induced by $S\mathbin{\dot\cup}T$. To keep $\sigma$ a fixed fraction of $d(S,T)$, we therefore need to bound $B_\infty$ independently of $\Delta$.

\subsection{Vertex peeling}
\label{sec:vertex-peeling}

For each multiplication, we select edges by peeling: we repeatedly remove a vertex of low degree, selecting its remaining edges if it is on the side of the next vector, so that its column sum is small, and discarding them if it is on the side of the current vector, so that only a few of its edges are discarded. For each target vertex in the support, most edges into the opposite target set are selected for multiplication. The thresholds for low degree are fractions of the degrees $a,b$ guaranteed by cohesiveness.

\begin{lemma}[Degree guesses for cohesive sets]
\label{lem:degree-guesses}
For cohesive sets $S\subseteq L$, $T\subseteq R$, there are powers of two $a,b\ge1$, depending only on $S,T$, such that $e(u,T)\ge a$ for every $u\in S$, $e(S,w)\ge b$ for every $w\in T$, and $\sqrt{ab}>d(S,T)/4$.
For every $k$ with $|S|,|T|\le k$, $a,b\le k$, $a|S|\le2bk$ and $b|T|\le2ak$.
\end{lemma}
The proof is in \Cref{app:cohesive-sets}. Since $a,b$ depend on unknown $S,T$, the algorithm tries all pairs of powers of two in $[1,k]$.

For a step from $L$ to $R$ with $U\coloneq\supp(x)$, the \defemph{neighborhood graph} of $U$ is $H\coloneq(U\mathbin{\dot\cup}W,F)$, where $W\coloneq N_G(U)$ is the set of neighbors of $U$ and $F$ is the set of edges incident to $U$. (For a step from $R$ to $L$, we argue in the same way after exchanging $a$ and $b$.) We seek a set $F'\subseteq F$ of selected edges such that every $w\in W$ is incident to at most $b/2$ selected edges and every $u\in U$ to at most $a/4$ discarded edges. With each selected edge assigned weight $2/b$, the weighted matrix has column-sum bound $B_\infty=1$. Each $u\in S\cap U$ has at least $a$ edges into $T$, of which at most $a/4$ are discarded, so at least $3a/4$ of its edges into $T$ belong to $F'$, regardless of its degree in $G$. Hence
\begin{equation}
\label{eq:local-growth}
\sum_{w\in T}y(w)\ge\frac{3a}{2b}\sum_{u\in S}x(u).
\end{equation}
Here $\sum_{u\in S\cap U}x(u)=\sum_{u\in S}x(u)$ since $U=\supp(x)$.
$\Peel$ (\Cref{alg:peel}) selects edges by removing low-degree vertices, or returns dense subsets of the vertices it cannot remove.

\begin{algorithm}[t]
\caption{\texorpdfstring{$\Peel(H,x,a,b)$}{Peel}: select edges or return subsets}
\label{alg:peel}
\begin{algorithmic}[1]\small
\Require A bipartite graph $H=(U\mathbin{\dot\cup}W,F)$, a nonnegative vector $x$ indexed by $U$, and $a,b>0$.
\Ensure A nonnegative vector $y$ indexed by $W$, or sets $U_*,Y$ with $\varnothing\ne U_*\subseteq U$ and $\varnothing\ne Y\subseteq W$.
\State $U_*\gets U$, \quad $W_*\gets W$, \quad $F'\gets\varnothing$ \label{line:peel-init}
\While{there exists $u\in U_*$ with $e_H(u,W_*)\le a/4$ or $w\in W_*$ with $e_H(U_*,w)\le b/2$} \label{line:peel-while}
  \State Let $q$ be the vertex with smallest identifier satisfying either condition. \label{line:peel-pick}
  \If{$q\in W_*$} \label{line:peel-ifW}
    \State $F'\gets F'\cup\{uq\in F : u\in U_*\}$ and $W_*\gets W_*\setminus\{q\}$ \label{line:peel-select}\label{line:peel-removeW}
  \Else
    \State $U_*\gets U_*\setminus\{q\}$ \label{line:peel-discard}
  \EndIf
\EndWhile
\State \algorithmicif\ $U_*\mathbin{\dot\cup}W_*=\varnothing$ \algorithmicthen\ \Return $y$ indexed by $W$ with $y(w)\coloneq\frac2b\sum_{u:\,uw\in F'}x(u)$ \label{line:peel-empty}\label{line:peel-vector}
\State Let $Y$ consist of the first $\min\!\left\{|W_*|,\left\lceil\frac{a|U_*|}{2b}\right\rceil\right\}$ vertices of $W_*$ in identifier order. \label{line:peel-Y}
\State \Return $(U_*,Y)$ \label{line:peel-sets}
\end{algorithmic}
\end{algorithm}

If $\Peel$ removes every vertex, then an edge belongs to $F'$ exactly when its endpoint in $W$ is removed before its endpoint in $U$ (line~\ref{line:peel-select}). Thus every $w\in W$ has at most $b/2$ selected edges and every $u\in U$ at most $a/4$ discarded edges, since these are its remaining edges when it is removed (line~\ref{line:peel-while}).
If vertices remain, then every vertex of $U_*$ has more than $a/4$ neighbors in $W_*$ and every vertex of $W_*$ more than $b/2$ neighbors in $U_*$ (line~\ref{line:peel-while}), so that $d_H(U_*,W_*)>\sqrt{ab/8}$. In the search, $|U_*|\le|\supp(x)|$, and the support size is bounded by pruning as in \Cref{eq:ov-size}, but $|W_*|$ is not. To bound the output size, we replace $W_*$ by a subset $Y$ (line~\ref{line:peel-Y}) with $|Y|\le a|U_*|/(2b)+1$. Every vertex of $Y$ still has more than $b/2$ neighbors in $U_*$, and the density remains above $\sqrt{ab/8}$ (see \Cref{lem:step-outputs}(b)).

\subsection{The algorithm and its guarantees}
\label{sec:algorithm-guarantee}
\label{sec:algorithm}

\begin{algorithm}[t]
\caption{Local approximation, for $p>2$}
\label{alg:local}
\begin{algorithmic}[1]\small
\Statex \hspace*{-\leftmargin}\textbf{$\Local_p(G,v,k)$}
\Require Neighbor-query access to $G=(L\mathbin{\dot\cup}R,E)$, a vertex $v\in L\mathbin{\dot\cup}R$, and an integer $k\ge1$.
\Ensure Nonempty $X\subseteq L$, $Y\subseteq R$, or $\bot$ if $v$ is isolated.
\State Query the first neighbor $u$ of $v$. \label{line:local-first}
\State \algorithmicif\ $u=\bot$ \algorithmicthen\ \Return $\bot$ \label{line:local-isolated}
\State \algorithmicif\ $v\in R$ \algorithmicthen\ exchange the roles of $L$ and $R$ below, and swap the returned sets \label{line:local-swap}
\State $\mathcal C\gets\bigl\{(\{v\},\{u\})\bigr\}$ \label{line:local-init}
\For{each $a,b\in\{2^0,2^1,\ldots,2^{\lfloor\log_2k\rfloor}\}$} \label{line:local-guess}
  \State $r\gets\Search_p(G,v,k,a,b)$ \label{line:local-search}
  \State \algorithmicif\ $r\ne\bot$ \algorithmicthen\ add $r$ (ordered with its current $L$-subset first) to $\mathcal C$ \label{line:local-add}
\EndFor
\State \Return $(X,Y)\in\mathcal C$ maximizing $d(X,Y)$ \label{line:local-return}
\Statex \smallskip\hspace*{-\leftmargin}\textbf{$\Search_p(G,v,k,a,b)$}
\Require The inputs above, with $v\in L$, and a guess $a,b\in[1,k]$.
\Ensure Two nonempty vertex sets on opposite sides of $G$, or $\bot$.
\State Take $\sigma,\delta,(K_t)_{t\ge0}$ from \Cref{eq:parameters}, and set $x_0\gets\one_v$. \label{line:search-init}
\For{$t=0,1,2,\ldots$} \label{line:search-for}
  \State $(a_t,b_t)\gets(a,b)$ if $t$ is even, and $(b,a)$ if $t$ is odd \label{line:search-parity}
  \For{each $u\in\supp(x_t)$} \label{line:search-scan}
    \State Read the adjacency list of $u$, stopping after $\lceil\sigma^2\rceil$ entries, and let $N_u$ be the set of neighbors read. \label{line:search-list}
    \State \algorithmicif\ $|N_u|=\lceil\sigma^2\rceil$ \algorithmicthen\ \Return $(\{u\},N_u)$ \label{line:search-star}
  \EndFor
  \State Let $H\coloneq(\supp(x_t)\mathbin{\dot\cup}W,F)$ be the neighborhood graph of $\supp(x_t)$, with $W\coloneq N_G(\supp(x_t))$. \label{line:search-H}
  \State $r\gets\Peel(H,x_t,a_t,b_t)$ \label{line:search-peel}
  \State \algorithmicif\ $r$ consists of sets $U_*,Y$ \algorithmicthen\ \Return $(U_*,Y)$ \label{line:search-sets}
  \State $y_{t+1}\gets r$, \quad $z\gets\prune(\delta,y_{t+1})$ \label{line:search-prune}
  \State \algorithmicif\ $z=0$ \algorithmicthen\ \Return $\bot$ \label{line:search-zero}
  \If{$\|z\|_{p,\infty}\le K_{t+1}^{1/p}$} \label{line:search-test}
    \State $x_{t+1}\gets z$ \label{line:search-accept}
  \Else
    \State Let $m\ge1$ be the smallest index such that $I_z(m)>K_{t+1}^{1/p}m^{1-1/p}$. \label{line:search-m}
    \State \Return $(P_i(x_t),P_m(z))$ for the smallest $i$ with $d_H(P_i(x_t),P_m(z))>\sigma$, or $\bot$ if none exists \label{line:search-sweep}
  \EndIf
\EndFor
\end{algorithmic}
\end{algorithm}

$\Local_p$ (\Cref{alg:local}) keeps the edge from $v$ to its first neighbor as the first candidate, so that the output is nonempty with density at least one (lines~\ref{line:local-first} to~\ref{line:local-init}). By guessing $a,b$ among the powers of two in $[1,k]$ (line~\ref{line:local-guess}), we lose only a constant factor in the approximation.
To keep the computation local, $\Search_p$ reads at most $\lceil\sigma^2\rceil$ entries of each adjacency list and returns a star if it finds that many neighbors, since a star with that many leaves has density at least $\sigma$ (line~\ref{line:search-star}). It accepts the pruned vector as the next one only if it passes the \defemph{weak-norm test} of line~\ref{line:search-test}, analyzed below.

For a starting vertex in $L$ and a guess $(a,b)$, the search uses the density threshold $\sigma$, the pruning threshold $\delta$, and the weak-norm thresholds $K_t^{1/p}$, where
\begin{equation}
\label{eq:parameters}
\sigma\coloneq\frac{\sqrt{3ab}}{4C_p^{p/2}},\qquad
\delta\coloneq\frac{\max\{1,a/b\}}{4k},\qquad
K_t\coloneq\begin{cases}(3/4)^t,&t\text{ even},\\(a/b)(3/4)^t,&t\text{ odd}.
\end{cases}
\end{equation}
The density threshold $\sigma$ is a constant fraction of $\sqrt{ab}$, and $\sigma<\sqrt{ab/8}$ since $C_p^p>2$.

For a step from $L$ to $R$, write $x\coloneq x_t$ and $y\coloneq y_{t+1}=(2/b)xA$, where $A$ is the adjacency matrix of the selected edges and has column-sum bound $B_\infty=b/2$ (line~\ref{line:search-peel}). If $d(P_i(x),P_m(y))\le\sigma$ for every prefix $P_i(x)$, then, using \Cref{eq:local-summation}, rescaling by $2/b$, and substituting $\sigma$ from \Cref{eq:parameters}, we get
\begin{equation}
\label{eq:ov-selected-bound}
I_y(m)\le\Bigl(\frac{3a}{4b}\Bigr)^{1/p}\|x\|_{p,\infty}m^{1-1/p}.
\end{equation}

For $\delta>0$ and a nonnegative vector $y$, $\prune(\delta,y)$ is obtained from $y$ by setting every coordinate smaller than $\delta$ to zero, and the pruned vector is $z=\prune(\delta,y_{t+1})$ (line~\ref{line:search-prune}). For every prefix $P_m(z)$, $P_m(z)=P_m(y_{t+1})$ and $I_z(m)=I_{y_{t+1}}(m)$, so \Cref{eq:ov-selected-bound} holds with $z$ in place of $y$ as well. Every positive coordinate of $z$ is at least $\delta$, and the sum over any set $Q$ decreases by at most $\delta|Q|$ (\Cref{app:prefix}).

For cohesive sets $S\subseteq L$, $T\subseteq R$, a starting vertex $v\in S$, and the search with the guess $(a,b)$ from \Cref{lem:degree-guesses}, we consider the mass on the target side, multiplied by $b/a$ when that side is $T$. Whenever $\Peel$ returns a vector, its weighted target mass before pruning is at least $3/2$ times the previous value, by \Cref{eq:local-growth}. By \Cref{lem:degree-guesses}, the bounds $|S|,|T|\le k$, and the choice of $\delta$ in \Cref{eq:parameters},
\begin{equation*}
|S|\le\min\{k,2bk/a\}\le\frac1{2\delta},\qquad
(b/a)|T|\le\min\{(b/a)k,2k\}\le\frac1{2\delta}.
\end{equation*}
Thus at most $\max\{\delta|S|,(b/a)\delta|T|\}\le1/2$ is lost during pruning. The initial weighted target mass is one, so by induction it remains at least one after every pruning step in the search with guess $(a,b)$. In particular, $z\ne0$ at each such step.

The search accepts $z$ as the next vector only if $\|z\|_{p,\infty}\le K_{t+1}^{1/p}$, so every $x_t$ satisfies $\|x_t\|_{p,\infty}\le K_t^{1/p}$ (lines~\ref{line:search-test} and~\ref{line:search-accept}). The thresholds are chosen to match \Cref{eq:ov-selected-bound}: starting from $K_0=1$, $K_t^{1/p}$ is multiplied by $(3a/(4b))^{1/p}$ in each step from $L$ to $R$ and by $(3b/(4a))^{1/p}$ in each step from $R$ to $L$, as specified for $K_t$ in \Cref{eq:parameters}. Hence, if $d(P_i(x_t),P_m(z))\le\sigma$ for all prefixes $P_i(x_t),P_m(z)$, then, by \Cref{eq:ov-selected-bound} and $\|x_t\|_{p,\infty}\le K_t^{1/p}$, $I_z(m)\le K_{t+1}^{1/p}m^{1-1/p}$ for every integer $1\le m\le|\supp(z)|$. By \Cref{lem:prefix}(a), the supremum defining $\|z\|_{p,\infty}$ is attained at such an integer, so $z$ passes the test. Over two steps, the threshold decreases by a factor of $(3/4)^{2/p}$.
If $z$ fails the test, $\Search_p$ finds the shortest prefix $P_m(z)$ whose sum exceeds $K_{t+1}^{1/p}m^{1-1/p}$ (line~\ref{line:search-m}). Since the bound \Cref{eq:ov-selected-bound} fails for this prefix, some prefix $P_i(x_t)$ has $d(P_i(x_t),P_m(z))>\sigma$, and the search returns these two prefixes (line~\ref{line:search-sweep}).

\paragraph{Output size and locality.}
Every positive coordinate of $x_t$ is at least $\delta$ and $\|x_t\|_{p,\infty}\le K_t^{1/p}$, so, by the choice of parameters in \Cref{eq:parameters},
\begin{equation}
\label{eq:ov-size}
|\supp(x_t)|\le K_t\delta^{-p}\le(4k)^p(3/4)^t,
\qquad \sum_t|\supp(x_t)|\le4(4k)^p,
\end{equation}
and each set returned by a sweep or by peeling has at most $(4k)^p$ vertices (\Cref{lem:call}(d)), so the output has $O_p(k^p)$ vertices. Every vertex of $\supp(x_{t+1})$ has a neighbor in $\supp(x_t)$, so the sets returned at step $t$ lie within distance $t+1$ of $v$, where $t\le p\log_{4/3}(4k)$ by \Cref{eq:ov-size}. In particular, for $p=2+\epsilon$, every returned vertex lies within distance $O_\epsilon(\log_2(2k))$ of $v$ (\Cref{lem:distance}).

The algorithm makes $O_p\!\left(k^p\min\{(\Delta+1)\log_2^2(2k),\ k^2\}\right)$ neighbor queries (\Cref{app:computation}).

\paragraph{Constant approximation.}
We prove that, for every $p>2$ and $S,T,k,v$ as in \Cref{thm:overview}, the output $(X,Y)$ of $\Local_p(G,v,k)$ satisfies
\begin{equation}
\label{eq:approximation}
d(X,Y)>\frac{\sqrt3}{16}\,C_p^{-p/2}\,d(S,T).
\end{equation}
We choose $(a,b)$ as in \Cref{lem:degree-guesses} and suppose first that $v\in S$. By the induction above, $z\ne0$ at every pruning step in the search with guess $(a,b)$. The search cannot continue indefinitely, since $1\le\sum_{u\in S}x_t(u)\le(3/4)^{t/p}k^{1-1/p}$ at every even step. Thus $\Search_p$ returns sets of density at least $\sigma$ within $O_p(\log_2(2k))$ steps. Since $\sqrt{ab}>d(S,T)/4$, their density satisfies \Cref{eq:approximation}. For $v\in T$, we exchange $L,R$ and $a,b$. All other searches also terminate by \Cref{eq:ov-size}, and $\Local_p$ returns the densest candidate.
For $\epsilon>0$ and $p\coloneq2+\epsilon$, we obtain the guarantees of \Cref{thm:overview}, with approximation ratio $O(1+1/\epsilon)$, by direct calculation from \Cref{eq:approximation,eq:ov-size}, using $C_p^{-p/2}\ge\epsilon/(\sqrt{\mathrm e}(1+\epsilon))$; see \Cref{app:quadratic-output} for details.

\section{Limitations on improving the output size}
\label{sec:limits}

For a target with at most $k$ vertices per side, our algorithm returns $O_\epsilon(k^{2+\epsilon})$ vertices with a constant approximation ratio. Can the output be reduced to $O(k)$ or $o(k^2)$ while keeping the approximation ratio constant? The threshold $k^2$ appears because we ask for approximation rather than recovery: a star with $k^2$ leaves has density $k$ and is as dense as any target with at most $k$ vertices per side, whereas a star with $o(k^2)$ leaves has density $o(k)$. We rule out output size $\phi(k)=o(k^2)$ per side in two settings:
\begin{itemize}[leftmargin=*,nolistsep]
\item \emph{Restricted queries.} If the worst-case number of queries is bounded by a function of $k$ alone, or more generally is $o(\Delta)$ for every fixed $k$, while computation between queries is unrestricted, then, on some instance and from every target vertex, even a randomized algorithm fails with probability at least $2/3$ to return a constant approximation. We choose $k$ centers per side uniformly among many disjoint stars and place the edges of a complete bipartite subgraph at uniformly chosen adjacency-list positions. We compare runs before and after this modification using the same algorithmic randomness. Unless the algorithm receives different answers or includes a target center other than the starting vertex in its output, any returned sets have density at most $\sqrt{2\phi(k)}=o(k)$, whereas the target has density $k$. By summing the success probabilities over target vertices and averaging over the random construction, we obtain one graph on which the algorithm fails with probability at least $2/3$ from every target vertex. The formal statement is \Cref{thm:queries}, proved in \Cref{app:queries}.

\item \emph{Restricted running time.} If the algorithm may read the entire graph but must run in time $f(k)|V|^{O(1)}$ for a computable function $f$ of $k$, the limitation is that we need to find not only a dense subgraph, but also a \emph{small} one. For a graph $J$ promised to contain a $k$-clique, the two copies $S,T$ of the clique in its bipartite double cover have $d(S,T)=k-1$. From returned sets $X,Y$ with $d(X,Y)=\Omega(k)$ and $|X|,|Y|\le\phi(k)$, we extract a set $U$ of $k$ vertices of $J$ spanning $\Omega\!\left({k^3}/{(k+\phi(k))}\right)$ edges.
Since the optimal $k$-vertex subgraph has $\binom{k}{2}$ edges, the approximation ratio for Densest $k$-Subgraph would be $O(1+\phi(k)/k)=o(k)$, which under the Strongish Planted Clique Hypothesis (SPCH) of \citet{ManurangsiRubinsteinSchramm2021} is impossible in this running time, even with randomization. The formal statement is \Cref{thm:quadratic}, proved in \Cref{app:spch}. Under ETH, we obtain a weaker output-size bound by the same reduction (\Cref{thm:limits}, proved in \Cref{app:eth}).
\end{itemize}

\section{Conclusion}
\label{sec:conclusion}

We give the first constant-factor local approximation algorithm for the bipartite densest subgraph problem, with output size, number of queries, and running time depending on $k$ alone (\Cref{thm:overview}). Like the density itself, the guarantee does not depend on edges outside the subgraph induced by $S\mathbin{\dot\cup}T$.

The main technical question left unanswered is the following: while \emph{sub}-quadratic output is impossible in the settings of \Cref{sec:limits}, quadratic output is not ruled out, and we leave open whether it is possible to achieve a constant-factor approximation with output of size $O(k^2)$.

Beyond improving the output size and approximation ratio, is it possible, under suitable assumptions, to ensure that the returned sets overlap with the target sets?  Such results are known in the case of conductance, i.e., the classical application of the Lov\'asz--Simonovits curve---see, for example, \citet{YangFountoulakis2023,BackDeLucaFountoulakisYang2024}.

\section*{Acknowledgements}
Guy Weissenberg was supported by the Swiss National Science Foundation grant 10004935 - Sublinear Algorithms for Non-linear Structures (Project MINT).

\section*{AI use statement}
In this work, we used generative AI tools for brainstorming, writing and editing. We have reviewed all AI-assisted work, and constructed the main arguments in the paper ourselves. We checked the LLM-generated ideas for potential plagiarism through a manual literature survey. We take responsibility for the final content of this work, including its presentation, text, claims and proofs.

\bibliographystyle{plainnat}
\bibliography{references}

\appendix

\section{Technical lemmas for the local algorithm}
\label{app:local-lemmas}

We prove the statements deferred from \Cref{sec:introduction,sec:analysis}.

\subsection{Cohesive subsets and degree guesses}
\label{app:cohesive-sets}

\begin{lemma}[Existence of cohesive subsets]
\label{lem:cohesive}
For all sets $S\subseteq L$ and $T\subseteq R$ with $e(S,T)>0$, there are cohesive subsets $S^*\subseteq S$ and $T^*\subseteq T$ satisfying $d(S^*,T^*)\ge d(S,T)$.
\end{lemma}

\begin{proof}
Choose nonempty sets $S^*\subseteq S$ and $T^*\subseteq T$ maximizing $d(S^*,T^*)$. Then $d(S^*,T^*)\ge d(S,T)>0$. For $u\in S^*$, the bound $e(u,T^*)\ge e(S^*,T^*)/(2|S^*|)$ is trivial if $S^*=\{u\}$. Otherwise, by maximality, $d(S^*\setminus\{u\},T^*)\le d(S^*,T^*)$, that is,
\begin{equation*}
\frac{e(S^*,T^*)-e(u,T^*)}{\sqrt{(|S^*|-1)|T^*|}}\le\frac{e(S^*,T^*)}{\sqrt{|S^*||T^*|}}
\enspace.
\end{equation*}
Rearranging and using $\sqrt{1-t}\le1-t/2$, we get
\begin{equation*}
1-\frac{e(u,T^*)}{e(S^*,T^*)}\le\sqrt{1-\frac1{|S^*|}}\le1-\frac1{2|S^*|}
\enspace,
\end{equation*}
so $e(u,T^*)\ge e(S^*,T^*)/(2|S^*|)$. By symmetry, $e(S^*,w)\ge e(S^*,T^*)/(2|T^*|)$ for every $w\in T^*$, so $S^*,T^*$ are cohesive.
\end{proof}

\begin{proof}[Proof of \Cref{lem:degree-guesses}]
Let $a_0\coloneq e(S,T)/(2|S|)$ and $b_0\coloneq e(S,T)/(2|T|)$. Since $S,T$ are cohesive, $e(u,T)\ge a_0$ for every $u\in S$ and $e(S,w)\ge b_0$ for every $w\in T$. The guesses $a$ and $b$ are the largest powers of two not exceeding $\max(a_0,1)$ and $\max(b_0,1)$, respectively, so that
\begin{equation}
\label{eq:rounding}
a_0<2a,\qquad b_0<2b
\enspace.
\end{equation}
Degrees are integers and $a_0,b_0>0$, so $e(u,T)\ge\max(a_0,1)\ge a$ for every $u\in S$, and similarly $e(S,w)\ge b$ for every $w\in T$. By \Cref{eq:rounding},
\begin{equation*}
\sqrt{ab}>\frac{\sqrt{a_0b_0}}2=\frac{d(S,T)}4
\enspace.
\end{equation*}

For any $k$ with $|S|,|T|\le k$, $a_0,b_0\le k/2$ and hence $a,b\le k$. It remains to show $a|S|\le2bk$, as $b|T|\le2ak$ follows by symmetry. If $a=1$, then $a|S|=|S|\le k\le2bk$. If $a>1$, then $a\le a_0$ and
\begin{equation*}
a|S|\le a_0|S|=b_0|T|<2b|T|\le2bk
\enspace.
\end{equation*}
\end{proof}

\begin{corollary}[Arbitrary target sets]
\label{cor:arbitrary}
For every $\epsilon>0$ and all sets $S\subseteq L$, $T\subseteq R$ with $e(S,T)>0$ and $|S|,|T|\le k$, there are $S'\subseteq S$ and $T'\subseteq T$ with $e(S',T')\ge e(S,T)/2$ such that, from every starting vertex $v\in S'\mathbin{\dot\cup}T'$, the algorithm $\Local_{2+\epsilon}(G,v,k)$ satisfies (i) of \Cref{thm:overview} and its output $(X,Y)$ satisfies $d(X,Y)\ge\frac{\sqrt3}{32\sqrt{\mathrm e}}\frac{\epsilon}{1+\epsilon}\,d(S,T)$.
\end{corollary}

\begin{proof}
By \Cref{lem:andersen-growth}, there are $S'\subseteq S$ and $T'\subseteq T$ with $e(S',T')\ge e(S,T)/2$ such that $e(u,T')>a_0$ for every $u\in S'$ and $e(S',w)>b_0$ for every $w\in T'$, where now $a_0\coloneq e(S,T)/(4|S|)$ and $b_0\coloneq e(S,T)/(4|T|)$. Since $e(S,T)\le|S||T|$, both $a_0$ and $b_0$ are at most $k/4$. Rounding as in the proof of \Cref{lem:degree-guesses}, we choose powers of two $a,b\le k$ with $e(u,T')\ge a$ for every $u\in S'$, $e(S',w)\ge b$ for every $w\in T'$, and $\sqrt{ab}>\sqrt{a_0b_0}/2=d(S,T)/8$.

If $a>1$, then $a|S'|\le a_0|S|=b_0|T|<2b|T|\le2bk$, and if $a=1$, then $a|S'|\le k\le2bk$. By symmetry, $b|T'|\le2ak$.

In the proof of \Cref{eq:approximation} in \Cref{sec:algorithm}, we use only these properties of the target sets and of the guess $(a,b)$, so from every $v\in S'\mathbin{\dot\cup}T'$ the output has density above $\frac{\sqrt3}{32}C_p^{-p/2}d(S,T)$. By \Cref{eq:exponent-factor}, $d(X,Y)>\frac{\sqrt3}{32\sqrt{\mathrm e}}\frac{\epsilon}{1+\epsilon}\,d(S,T)$. The bounds in (i) hold for every starting vertex by \Cref{thm:computation}.
\end{proof}

\subsection{\texorpdfstring{Weak-$\ell_p$}{Weak-lp} preliminaries and pruning}
\label{app:prefix}

\begin{lemma}[Prefix sums, support size, and pruning]
\label{lem:prefix}
For a nonzero vector $x\ge0$ with $\ell\coloneq|\supp(x)|$, $p\ge2$ and $\delta>0$:
\begin{enumerate}[label=(\alph*),nolistsep]
\item The supremum in the definition of $\|x\|_{p,\infty}$ is attained at an integer $m\le\ell$, that is,
\begin{equation}
\label{eq:weak-bound}
\|x\|_{p,\infty}=\max_{1\le m\le\ell}\frac{I_x(m)}{m^{1-1/p}}
\enspace.
\end{equation}
\item If every positive coordinate of $x$ is at least $\delta$, then $\ell\le\|x\|_{p,\infty}^p/\delta^p$.
\item For $y\ge0$ and $z\coloneq\prune(\delta,y)$, $z\le y$, every positive coordinate of $z$ is at least $\delta$, and
\begin{equation}
\label{eq:pruning-error}
0\le \sum_{w\in Q}\bigl(y(w)-z(w)\bigr)\le\delta|Q|\qquad(Q\subseteq V)
\enspace.
\end{equation}
\end{enumerate}
\end{lemma}

\Needspace{4\baselineskip}
\begin{proof}[Proof of \Cref{lem:prefix}]
We prove each item below.
\begin{enumerate}[label=(\alph*),leftmargin=*,itemsep=\smallskipamount,parsep=0pt]
\item On $(0,1]$ the ratio $I_x(s)/s^{1-1/p}=x_{(1)}s^{1/p}$ increases, and for $s>\ell$ the numerator is constant while the denominator increases. On each interval $[i-1,i]$ with $2\le i\le\ell$, $I_x$ is linear (\Cref{def:prefix}) and $s^{1-1/p}$ is concave. With $c$ the maximum of the ratio over this interval, $I_x(s)-cs^{1-1/p}$ is convex with maximum zero, and a convex function attains its maximum at an endpoint of the interval, so the ratio equals $c$ at an endpoint. Hence the supremum is attained at an integer $m\le\ell$.

\item Every positive coordinate is at least $\delta$, so $\ell\delta\le I_x(\ell)\le\|x\|_{p,\infty}\ell^{1-1/p}$ by \Cref{def:prefix}, and hence $\ell\le\|x\|_{p,\infty}^p/\delta^p$.

\item Each coordinate of $y-z$ lies in $[0,\delta)$. \qedhere
\end{enumerate}
\end{proof}

\begin{remark*}
The norm $\|x\|_{p,\infty}$ is the norm $\|x\|_{(p,\infty)}$ of \citet[Chapter~4]{BennettSharpley1988} on the weak-$\ell_p$ sequence space, the Lorentz space $\ell_{p,\infty}$. It lies between the quasinorm $\sup_i i^{1/p}x_{(i)}$ and $p/(p-1)$ times that quasinorm.
\end{remark*}

\subsection{The interpolation inequality}
\label{app:interpolation}

\begin{proposition}
\label{prop:interpolation}
For $p>2$, a vector $x\ge0$ with nonempty support $U$ and $\ell\coloneq|U|$, and a nonnegative $U\times W$ matrix $B$ with column-sum bound $B_\infty>0$ (\Cref{def:column-sum}), if $\gamma>0$ and a set $Q\subseteq W$ of size $m>0$ satisfy
\begin{equation}
\label{eq:prefix-hypothesis}
\sum_{u\in P_i(x)}\sum_{w\in Q}B_{uw}\le\gamma\sqrt{im}
\qquad(1\le i\le\ell)
\enspace,
\end{equation}
then
\begin{equation}
\label{eq:power-p}
\sum_{w\in Q}(xB)(w)
\le C_p\,\gamma^{2/p}B_\infty^{1-2/p}\|x\|_{p,\infty}m^{1-1/p}
\enspace.
\end{equation}
\end{proposition}

\begin{proof}[Proof of \Cref{prop:interpolation}]
We write $U=\{u_1,\ldots,u_\ell\}$ in the order of \Cref{def:prefix} and set $c_i\coloneq\sum_{w\in Q}B_{u_i,w}$, so that $\sum_{w\in Q}(xB)(w)=\sum_{i=1}^{\ell}x_{(i)}c_i$. By \Cref{eq:prefix-hypothesis} and the column-sum bound,
\begin{equation*}
\sum_{j\le i}c_j\le g(i)\qquad(1\le i\le\ell),
\qquad\text{where}\qquad
g(s)\coloneq\min\{\gamma\sqrt{ms},\,B_\infty m\}\quad(s\ge0)
\enspace.
\end{equation*}
Then $g(0)=0$, and the two terms defining $g$ are equal at $r\coloneq B_\infty^2m/\gamma^2$, so $g'(s)=\gamma\sqrt m/(2\sqrt s)$ for $0<s<r$ and $g'(s)=0$ for $s>r$. By summation by parts, with $x_{(\ell+1)}=0$,
\begin{equation*}
\sum_{i=1}^{\ell}x_{(i)}c_i
=\sum_{i=1}^{\ell}(x_{(i)}-x_{(i+1)})\sum_{j\le i}c_j
\le\sum_{i=1}^{\ell}(x_{(i)}-x_{(i+1)})g(i)
=\sum_{i=1}^{\ell}x_{(i)}\bigl(g(i)-g(i-1)\bigr)
\enspace.
\end{equation*}
Since $I_x$ is linear with slope $x_{(i)}$ on $[i-1,i]$ and constant beyond $\ell$ (\Cref{def:prefix}), we can write the last sum as an integral. Integrating by parts, with a zero boundary term at $0$ because $I_x(s)=x_{(1)}s$ for $0\le s\le1$, we obtain
\begin{equation*}
\begin{aligned}
\sum_{w\in Q}(xB)(w)
&\le\int_0^\infty I_x'(s)g'(s)\,ds
&=\frac{\gamma\sqrt m}2\int_0^r\frac{I_x'(s)}{\sqrt s}\,ds
&=\frac{\gamma\sqrt m}2\left[\frac{I_x(r)}{\sqrt r}
+\frac12\int_0^r\frac{I_x(s)}{s^{3/2}}\,ds\right]
\enspace.
\end{aligned}
\end{equation*}
By \Cref{def:prefix}, $I_x(s)\le\|x\|_{p,\infty}s^{1-1/p}$ for every real $s>0$. With $\tau\coloneq1/2-1/p>0$, the bracketed expression is at most
\begin{equation*}
\|x\|_{p,\infty}\left(r^\tau+\frac12\int_0^rs^{\tau-1}\,ds\right)
=\|x\|_{p,\infty}r^\tau\Bigl(1+\frac1{2\tau}\Bigr)
=2C_p\,\|x\|_{p,\infty}r^\tau
\enspace.
\end{equation*}
Finally
\begin{equation*}
\gamma\sqrt m\,r^\tau
=\gamma\sqrt m\Bigl(\frac{B_\infty^2m}{\gamma^2}\Bigr)^{1/2-1/p}
=\gamma^{2/p}B_\infty^{1-2/p}m^{1-1/p}
\enspace,
\end{equation*}
and \Cref{eq:power-p} holds.
\end{proof}

The restriction $p>2$ is needed only for $\tau>0$, so that the integral of $s^{\tau-1}$ converges at zero. For $p=2$, the integrand is $1/s$, whose integral over $[1,r]$ is $\ln r$ for $r\ge1$, and we must bound the small prefixes separately, as in \Cref{app:log-factor}.

\paragraph{The prefix-sum bound.}
The bound \Cref{eq:local-summation} in \Cref{sec:prefix-sweeps} is the case of \Cref{prop:interpolation} in which $B$ is the adjacency matrix $A$ of a subgraph of $G$ with column-sum bound $B_\infty$, $x\ge0$ is supported on $L$, $y=xA$, and $Q=P_m(y)$ is a prefix of $y$. If $d(P_i(x),P_m(y))\le\sigma$ for every prefix $P_i(x)$, then
\begin{equation*}
\sum_{u\in P_i(x)}\sum_{w\in P_m(y)}A_{uw}\le e(P_i(x),P_m(y))\le\sigma\sqrt{im}
\qquad(1\le i\le\ell)
\enspace,
\end{equation*}
which is \Cref{eq:prefix-hypothesis} with $\gamma=\sigma$, and \Cref{eq:power-p} is \Cref{eq:local-summation}.

\subsection{The one-step search bound}
\label{app:step}

\begin{lemma}[Prefix sums after peeling]
\label{lem:step}
For $p>2$, $a,b,\sigma>0$, a vector $x\ge0$ with nonempty support $U\subseteq L$, the neighborhood graph $H=(U\mathbin{\dot\cup}W,F)$ of $U$ with $W\coloneq N_G(U)$, and a vector $y$ returned by $\Peel(H,x,a,b)$, every nonempty $Q\subseteq W$ with $d(P_i(x),Q)\le\sigma$ for all $1\le i\le|U|$ satisfies
\begin{equation}
\label{eq:one-step}
\sum_{w\in Q}y(w)\le\left(\frac{4C_p^p\sigma^2}{b^2}\right)^{1/p}\|x\|_{p,\infty}|Q|^{1-1/p}
\enspace.
\end{equation}
The same bound holds for $z\coloneq\prune(\delta,y)$, for every $\delta>0$.
For $U\subseteq R$, the same holds with $a$ and $b$ exchanged, so that $a^2$ replaces $b^2$ in \Cref{eq:one-step}.
\end{lemma}

\begin{proof}[Proof of \Cref{lem:step}]
We write $F'$ for the selected edges and $B$ for the $U\times W$ matrix with $B_{uw}\coloneq2/b$ for $uw\in F'$ and $B_{uw}\coloneq0$ otherwise, so that $y=xB$. Since $\Peel$ selects at most $b/2$ edges incident to each $w\in W$,
\begin{equation*}
\sum_{u\in U}B_{uw}\le1\qquad(w\in W)
\enspace.
\end{equation*}
For every prefix $P_i(x)$,
\begin{equation*}
\sum_{u\in P_i(x)}\sum_{w\in Q}B_{uw}
\le\frac2b\,e(P_i(x),Q)
\le\frac{2\sigma}b\sqrt{i|Q|}
\enspace.
\end{equation*}
We apply \Cref{prop:interpolation} with $B_\infty=1$ and $\gamma=2\sigma/b$ to obtain \Cref{eq:one-step}. Since $0\le z\le y$, the same bound holds for $z$.
\end{proof}

With $\sigma$ from \Cref{eq:parameters} and $Q=P_m(y)$, \Cref{eq:one-step} is \Cref{eq:ov-selected-bound}.

\begin{lemma}[Scanning and returned sets]
\label{lem:step-outputs}
The following holds for one step of $\Search_p$ with peeling parameters $a,b>0$ and density threshold $\sigma>0$ from a nonnegative vector $x$ with nonempty support $U\subseteq L$, where $W\coloneq N_G(U)$.
\begin{enumerate}[label=(\alph*),nolistsep]
\item If no star is returned, then the adjacency lists read in line~\ref{line:search-list} contain every edge of the neighborhood graph $H$ of $U$ (line~\ref{line:search-H}), and $\deg(u)<\sigma^2$ for every $u\in U$.
\item If $\Peel$ returns sets $U_*\subseteq U$ and $Y\subseteq W$, both are nonempty and
\begin{equation*}
|Y|\le\frac{a|U_*|}{2b}+1,\qquad d(U_*,Y)>\sqrt{ab/8}
\enspace.
\end{equation*}
\end{enumerate}
For $U\subseteq R$, the same conclusions hold with $a$ and $b$ exchanged.
\end{lemma}

\Needspace{4\baselineskip}
\begin{proof}[Proof of \Cref{lem:step-outputs}]
We prove each item below.
\begin{enumerate}[label=(\alph*),leftmargin=*,itemsep=\smallskipamount,parsep=0pt]
\item If no star is returned, every adjacency list of a vertex in $U$ is read to its end (line~\ref{line:search-star}). Thus these lists contain every edge incident to $U$, and $\deg(u)\le\lceil\sigma^2\rceil-1<\sigma^2$ for every $u\in U$.

\item $\Search_p$ calls $\Peel$ only if no star is returned, so by (a), $d_H(P,Q)=d(P,Q)$ for all $P\subseteq U$ and $Q\subseteq W$. Both $U_*$ and $W_*$ are nonempty, since an isolated vertex would have been removed. $Y$ is $W_*$ or consists of $\lceil a|U_*|/(2b)\rceil$ vertices of $W_*$ (line~\ref{line:peel-Y}), so $Y$ is nonempty and
\begin{equation*}
|Y|\le\left\lceil\frac{a|U_*|}{2b}\right\rceil\le\frac{a|U_*|}{2b}+1
\enspace.
\end{equation*}

If $Y=W_*$, then $e_H(U_*,Y)>(a/4)|U_*|$ and $e_H(U_*,Y)>(b/2)|Y|$ by the loop condition (line~\ref{line:peel-while}), so
\begin{equation*}
d_H(U_*,Y)^2=\frac{e_H(U_*,Y)}{|U_*|}\cdot\frac{e_H(U_*,Y)}{|Y|}>\frac a4\cdot\frac b2=\frac{ab}8
\enspace.
\end{equation*}

If $Y\ne W_*$, then $|Y|=\lceil a|U_*|/(2b)\rceil$, and every vertex of $Y$ has more than $b/2$ neighbors in $U_*$, so
\begin{equation*}
d_H(U_*,Y)>\frac b2\sqrt{\frac{|Y|}{|U_*|}}
\ge\frac b2\sqrt{\frac a{2b}}=\sqrt{ab/8}
\enspace.\qedhere
\end{equation*}
\end{enumerate}
\end{proof}

\section{Completing the analysis of the local algorithm}
\label{app:local}

We complete the proof of \Cref{thm:overview} and prove \Cref{cor:exponent}.

\subsection{Search invariants and termination}
\label{app:call}

\begin{lemma}
\label{lem:call}
For $p>2$, $k\ge1$, $a,b\in[1,k]$ and a vertex $v\in L$ of positive degree, $\Search_p(G,v,k,a,b)$ with the parameters in \Cref{eq:parameters} satisfies the following.
\begin{enumerate}[label=(\alph*),nolistsep]
\item Every vector $x_t$ satisfies $\|x_t\|_{p,\infty}\le K_t^{1/p}$.
\item Every positive coordinate of $x_t$ is at least $\delta$, and $|\supp(x_t)|\le K_t/\delta^p\le(4k)^p(3/4)^t$.
\item The search terminates and never returns $\bot$ from the sweep.
\item In each step, the search reads at most $\lceil\sigma^2\rceil\le ab$ entries of each adjacency list. Whenever the search returns sets $X,Y$, they satisfy $d(X,Y)\ge\sigma$ and $|X|,|Y|\le(4k)^p$.
\end{enumerate}
For $v\in R$, the same holds for the search with $L$ and $R$ exchanged.
\end{lemma}

\Needspace{4\baselineskip}
\begin{proof}[Proof of \Cref{lem:call}]
We prove each item below.
\begin{enumerate}[label=(\alph*),leftmargin=*,itemsep=\smallskipamount,parsep=0pt]
\item $\|x_0\|_{p,\infty}=1=K_0^{1/p}$, and every later $x_t$ satisfies the weak-norm test $\|x_t\|_{p,\infty}\le K_t^{1/p}$ (line~\ref{line:search-test}).

\item Positive coordinates are at least $\delta$ at $t=0$, since $x_0=\one_v$ and $\delta=\max\{1,a/b\}/(4k)\le1/4$, and at later steps by pruning (\Cref{lem:prefix}(c)). By (a) and \Cref{lem:prefix}(b), $|\supp(x_t)|\le\|x_t\|_{p,\infty}^p/\delta^p\le K_t/\delta^p$. Since $\delta\ge1/(4k)$ and $K_t\le4k\delta(3/4)^t$ by \Cref{eq:parameters},
\begin{equation}
\label{eq:support-bound}
\frac{K_t}{\delta^p}
\le4k\delta^{1-p}\Bigl(\frac34\Bigr)^t
\le(4k)^p\Bigl(\frac34\Bigr)^t
\enspace.
\end{equation}

\item Every vector $x_t$ is nonzero (line~\ref{line:search-zero}), so, by (b), $1\le|\supp(x_t)|\le(4k)^p(3/4)^t$. The upper bound tends to zero, so the search computes only finitely many vectors $x_t$ and terminates.

For the second claim of (c), we show that in the sweep at step $t$ we find $i$ with $d(P_i(x_t),P_m(z))>\sigma$, where $z=\prune(\delta,y_{t+1})$ and $P_m(z)$ is the prefix of line~\ref{line:search-m}. Otherwise, for even $t$, by \Cref{lem:step} with $Q=P_m(z)$ and by (a),
\begin{equation*}
I_z(m)\le\left(\frac{4C_p^p\sigma^2}{b^2}\right)^{1/p}K_t^{1/p}m^{1-1/p}
\enspace.
\end{equation*}
By \Cref{eq:parameters}, $4C_p^p\sigma^2=3ab/4$, so the factor is $(3a/(4b))^{1/p}=(K_{t+1}/K_t)^{1/p}$. For odd $t$, \Cref{eq:one-step} holds with $a^2$ in place of $b^2$, and the factor is $(3b/(4a))^{1/p}=(K_{t+1}/K_t)^{1/p}$. In both cases, $I_z(m)\le K_{t+1}^{1/p}m^{1-1/p}$, contradicting the choice of $m$. The graph $H$ is the neighborhood graph of $\supp(x_t)$ (line~\ref{line:search-H}), so $d_H(P_i(x_t),P_m(z))=d(P_i(x_t),P_m(z))>\sigma$, and the sets $P_i(x_t),P_m(z)$ are returned by the sweep.

\item In each step, the search reads at most $\lceil\sigma^2\rceil$ entries of each adjacency list (line~\ref{line:search-list}). By Bernoulli's inequality, $C_p^p=(1+\frac1{p-2})^p\ge1+\frac p{p-2}>2$, so $\sigma^2=3ab/(16C_p^p)<ab/8$, and since $ab\ge1$, it follows that $\lceil\sigma^2\rceil\le\lceil ab/8\rceil\le ab\le k^2\le(4k)^p$. A returned star therefore has at most $(4k)^p$ leaves, and its density is $\sqrt{\lceil\sigma^2\rceil}\ge\sigma$. The sets returned by the sweep have density above $\sigma$ (line~\ref{line:search-sweep}), and the sets returned by peeling have density above $\sqrt{ab/8}>\sigma$ (\Cref{lem:step-outputs}(b)).

For sets $P_i(x_t),P_m(z)$ returned by the sweep at step $t$, the prefix $P_i(x_t)$ lies in $\supp(x_t)$, so $i\le(4k)^p$ by (b). By \Cref{lem:prefix}(a), $m\le|\supp(z)|$, so $z_{(m)}\ge\delta$ by pruning, and since $m$ is the smallest index with $I_z(m)>K_{t+1}^{1/p}m^{1-1/p}$,
\begin{equation*}
(m-1)\delta\le I_z(m-1)\le K_{t+1}^{1/p}(m-1)^{1-1/p}
\enspace,
\end{equation*}
so, by \Cref{eq:support-bound} and $(4k)^p\ge16$,
\begin{equation*}
m\le\frac{K_{t+1}}{\delta^p}+1
\le\frac34(4k)^p+1\le(4k)^p
\enspace.
\end{equation*}

For sets $U_*,Y$ returned by peeling at step $t$, the set $U_*$ lies in $\supp(x_t)$, so $|U_*|\le K_t/\delta^p$ by (b). By \Cref{lem:step-outputs}(b), $|Y|\le a|U_*|/(2b)+1$ for even $t$, and $|Y|\le b|U_*|/(2a)+1$ for odd $t$. In both cases, the coefficient of $|U_*|$ is $2K_{t+1}/(3K_t)$ by \Cref{eq:parameters}. Hence, by \Cref{eq:support-bound},
\begin{equation*}
|Y|\le\frac{2K_{t+1}}{3K_t}|U_*|+1
\le\frac{2K_{t+1}}{3\delta^p}+1
\le\frac{(4k)^p}2+1\le(4k)^p
\enspace.\qedhere
\end{equation*}
\end{enumerate}
\end{proof}

\subsection{Query complexity, running time, and locality}
\label{app:computation}

\paragraph{Graph access.}
\label{sec:graph-access}
The algorithm accesses $G$ through neighbor queries. Each adjacency list has a fixed order. When the algorithm makes a query $(u,i)$ for an integer $i\ge1$, it receives the $i$th neighbor of $u$, or $\bot$ if $i>\deg(u)$, where $\deg(u)$ is the degree of $u$. To read the list to its end, the algorithm makes $\deg(u)+1$ queries. We assume that the side of a vertex can be read from its identifier. We count neighbor queries and elementary operations at unit cost, including reads, copies, and comparisons of vertex identifiers. For $\Theta(\log|V|)$-bit identifiers, reading an identifier takes $\Theta(\log|V|)$ bit operations.

\begin{theorem}[Local computation]
\label{thm:computation}
For $p>2$, $k\ge1$ and every starting vertex $v$, the algorithm $\Local_p(G,v,k)$ terminates. If $v$ is isolated, it returns $\bot$ after one query. Otherwise, the output has at most $2(4k)^p$ vertices and density at least one, the number of neighbor queries is at most
\begin{equation}
\label{eq:queries}
1+4(4k)^p\min\left\{(\Delta+1)(1+\lfloor\log_2k\rfloor)^2,\ 4k^2\right\}
\enspace,
\end{equation}
and the running time is
\begin{equation}
\label{eq:time}
O_p\!\left(\min\left\{\Delta k^p\log_2^3(2k),\
k^{p+2}\log_2(2k)\right\}\right)
\enspace.
\end{equation}
\end{theorem}

\begin{proof}[Proof of \Cref{thm:computation}]
If $v$ is isolated, the first query is answered with $\bot$, and $\Local_p$ returns $\bot$ as well (line~\ref{line:local-isolated}). Otherwise, every search terminates by \Cref{lem:call}(c). By \Cref{lem:call}(d), the sets returned by the searches have at most $(4k)^p$ vertices per side, so the output has at most $2(4k)^p$ vertices. Its density is at least one because the initial edge is a candidate.

By \Cref{lem:call}(b), for each guess $(a,b)$, the algorithm scans at most $\sum_t|\supp(x_t)|\le4(4k)^p$ adjacency lists in total, with at most $\min\{\Delta+1,\lceil\sigma^2\rceil\}\le\min\{\Delta+1,ab\}$ queries per list by \Cref{lem:call}(d). Including the initial query, the total is at most
\begin{equation*}
1+4(4k)^p\sum_{a,b}\min\{\Delta+1,ab\}
\enspace,
\end{equation*}
where there are $(1+\lfloor\log_2k\rfloor)^2$ guesses and
\begin{equation*}
\sum_{a,b}ab=\left(\sum_{j=0}^{\lfloor\log_2k\rfloor}2^j\right)^2<4k^2
\enspace.
\end{equation*}
Thus the number of queries is at most \Cref{eq:queries}.

For the running time, we consider a step in which no star is returned, with the neighborhood graph $H=(U\mathbin{\dot\cup}W,F)$ of $U\coloneq\supp(x_t)$. Every vertex of $U$ has a neighbor, since $v$ has positive degree and every vertex of $\supp(x_{t+1})$ is joined to $\supp(x_t)$ by a selected edge (line~\ref{line:peel-vector}), so $|U|+|W|\le2|F|$. By \Cref{lem:call}(b), $|U|\le(4k)^p$, and by \Cref{lem:step-outputs}(a) and \Cref{eq:parameters}, $\deg(u)<\sigma^2<ab\le k^2$ for every $u\in U$, so
\begin{equation*}
|F|=\sum_{u\in U}\deg(u)<k^2(4k)^p
\enspace.
\end{equation*}
The time spent in this step sorting the vectors, running $\Peel$ and performing any sweep is $O(|F|\log_2|F|)=O_p(|F|\log_2(2k))$, and the algorithm makes at least $|F|$ queries to build $H$. When a star is returned, the work in that step is linear in the number of queries. Summing over all steps and guesses, we obtain a time bound of $O_p(\log_2(2k))$ per query and hence \Cref{eq:time}.
\end{proof}

\begin{lemma}[Distance from the starting vertex]
\label{lem:distance}
For sets $X\subseteq L$ and $Y\subseteq R$ returned by $\Local_p(G,v,k)$, every vertex of $X\mathbin{\dot\cup}Y$ lies within distance $1+p\log_2(4k)/\log_2(4/3)$ of $v$ in $G$. If $v\in L$ then $d(X\cup\{v\},Y)\ge d(X,Y)/\sqrt2$, and if $v\in R$ then $d(X,Y\cup\{v\})\ge d(X,Y)/\sqrt2$.
\end{lemma}

\begin{proof}[Proof of \Cref{lem:distance}]
In every search, $\supp(x_0)=\{v\}$, and every vertex of $\supp(x_{t+1})$ has a neighbor in $\supp(x_t)$, since $y_{t+1}$ is indexed by $N_G(\supp(x_t))$ and $x_{t+1}\le y_{t+1}$. By induction, $\supp(x_t)$ lies within distance $t$ of $v$. At step $t$, the vertices of a returned star, of the sets returned by peeling and of the prefixes returned by the sweep lie in $\supp(x_t)\mathbin{\dot\cup}N_G(\supp(x_t))$ and hence within distance $t+1$ of $v$. The initial edge is within distance one of $v$. Every vector $x_t$ is nonzero, so $1\le|\supp(x_t)|\le(4k)^p(3/4)^t$ by \Cref{lem:call}(b), and thus $t\le p\log_2(4k)/\log_2(4/3)$.

For the last claim, if $v\in L$, then $e(X\cup\{v\},Y)\ge e(X,Y)$ and $|X\cup\{v\}|\le2|X|$, so $d(X\cup\{v\},Y)\ge d(X,Y)/\sqrt2$. The case $v\in R$ is symmetric.
\end{proof}

\subsection{Choosing the exponent}
\label{app:quadratic-output}

\emph{Proof of \Cref{thm:overview}.} For every $p>2$, $\Local_p(G,v,k)$ returns nonempty sets $X\subseteq L$ and $Y\subseteq R$ with
\begin{equation*}
d(X,Y)>\frac{\sqrt3}{16}\,C_p^{-p/2}\,d(S,T),
\qquad C_p=1+\frac1{p-2},
\end{equation*}
as proved in \Cref{sec:algorithm}. For $\epsilon>0$ and $p\coloneq2+\epsilon$, the constant $C_p$ equals $1+1/\epsilon$. By $\ln(1+1/\epsilon)\le1/\epsilon$,
\begin{equation}
\label{eq:exponent-factor}
C_p^{-p/2}
=\left(\frac{\epsilon}{1+\epsilon}\right)^{1+\epsilon/2}
=\frac{\epsilon}{1+\epsilon}\left(1+\frac1\epsilon\right)^{-\epsilon/2}
\ge\frac{\epsilon}{\sqrt{\mathrm e}(1+\epsilon)}
\enspace.
\end{equation}
Hence $d(X,Y)>\frac{\sqrt3}{16\sqrt{\mathrm e}}\frac{\epsilon}{1+\epsilon}\,d(S,T)$, which is (ii) of \Cref{thm:overview}. By \Cref{thm:computation}, the output has $O_\epsilon(k^{2+\epsilon})$ vertices, and the algorithm makes $O_\epsilon(k^{2+\epsilon}\min\{\Delta\log_2^2(2k),k^2\})$ neighbor queries and runs in time $\widetilde O_\epsilon(k^{2+\epsilon}\min\{\Delta,k^2\})$, as claimed in \Cref{thm:overview}.

\begin{corollary}[Quadratic output size]
\label{cor:exponent}
For cohesive sets $S\subseteq L$ and $T\subseteq R$, $k\ge1$ with $|S|,|T|\le k$, a vertex $v\in S\mathbin{\dot\cup}T$, $\epsilon\coloneq1/\lceil\log_2(2k)\rceil$ and $p\coloneq2+\epsilon$, the algorithm $\Local_p(G,v,k)$ returns nonempty sets $X\subseteq L$ and $Y\subseteq R$ with $|X|+|Y|\le128k^2$ and
\begin{equation*}
d(X,Y)>\frac{\sqrt6}{128}\,\frac{d(S,T)}{\log_2(2k)}
\enspace.
\end{equation*}
For every starting vertex, the number of neighbor queries is at most
\begin{equation*}
1+256k^2\min\left\{(\Delta+1)(1+\lfloor\log_2k\rfloor)^2,\ 4k^2\right\}
\enspace.
\end{equation*}
\end{corollary}

\begin{proof}[Proof of \Cref{cor:exponent}]
The choice $\epsilon=1/\lceil\log_2(2k)\rceil$ satisfies $1/(2\log_2(2k))\le\epsilon\le1$ and $2^{1/\epsilon}\ge1+1/\epsilon$. By the identity in \Cref{eq:exponent-factor},
\begin{equation*}
C_p^{-p/2}\ge\frac{\epsilon}{2\sqrt2}\ge\frac1{4\sqrt2\log_2(2k)}
\enspace.
\end{equation*}
By this inequality and \Cref{eq:approximation}, $d(X,Y)>\frac{\sqrt6}{128}\,d(S,T)/\log_2(2k)$. Since $\log_2(2k)\ge1$, $\log_2(4k)=\log_2(2k)+1\le2\lceil\log_2(2k)\rceil$, and hence
\begin{equation*}
(4k)^\epsilon
=2^{\log_2(4k)/\lceil\log_2(2k)\rceil}
\le2^2=4
\enspace.
\end{equation*}
Thus $(4k)^p\le64k^2$, and by \Cref{thm:computation} there are at most $2(4k)^p\le128k^2$ output vertices. Substituting $(4k)^p\le64k^2$ into \Cref{eq:queries}, we get the query bound in \Cref{cor:exponent}.
\end{proof}

\section{\texorpdfstring{The logarithmic factor at $p=2$}{The logarithmic factor at p=2} and comparison with Andersen}
\label{app:andersen-prefix}

Andersen's guarantee has two parts: growth of the Euclidean norm inside a dense bipartite subgraph, and extraction of a dense bipartite subgraph from that growth (\Cref{app:andersen-lemmas}). After recalling his iteration and the two lemmas, we describe the neighbor-query implementation of his algorithm in \Cref{app:andersen-queries}. With the weak-$\ell_2$ bound $I_x(i)\le\sqrt i$ in place of the bound on $\|x\|_{p,\infty}$, we incur a logarithmic factor in the prefix-sum argument of \Cref{sec:prefix-sweeps} (\Cref{app:log-factor}), and there is a weighted matrix for which the logarithmic factor in the extraction is necessary (\Cref{app:extraction-example}).

\subsection{Andersen's growth and extraction lemmas}
\label{app:andersen-lemmas}

\paragraph{Andersen's iteration.}
Andersen's algorithm runs the pruned growth process \citep[Definition~5]{Andersen2010} for $t_0\coloneq\lceil\log_2(2\sqrt{k})\rceil$ steps. Starting from $x_0\coloneq\one_v$, with the pruning thresholds $\delta_{t+1}>0$ of that definition and the adjacency matrix $A$ of $G$, it computes
\begin{equation*}
y_{t+1}\coloneq\frac{x_tA}{\|x_tA\|_2},\qquad
x_{t+1}\coloneq\prune(\delta_{t+1},y_{t+1})
\enspace.
\end{equation*}
Each $y_{t+1}$ is a unit vector, so $\|x_t\|_2\le1$.

\begin{definition}[Dyadic level sets {[\citealp[Definition~6]{Andersen2010}]}]
\label{def:dyadic-level-sets}
For a nonnegative vector $x$, its \defemph{dyadic level sets} are
$D_j(x)\coloneq\{u:2^{-j-1}<x(u)\le2^{-j}\}$ for $j\in\mathbb Z$.
\end{definition}

At each step $0\le t<t_0$, Andersen's algorithm tests the nonempty sets
\begin{equation*}
X\coloneq D_i\!\left(\frac{x_t}{\|x_t\|_2}\right),\qquad
Y\coloneq D_j\!\left(\frac{x_tA}{\|x_tA\|_2}\right)
\quad\text{with }|i-j|\le4+\log_2\Delta
\enspace.
\end{equation*}
It returns the tested sets $X,Y$ of maximum density.

\begin{lemma}[Growth inside a dense bipartite subgraph {[\citealp[Lemma~2 and proof of Theorem~1]{Andersen2010}]}]
\label{lem:andersen-growth}
For $S\subseteq L$ and $T\subseteq R$ with $e(S,T)>0$ and $k$ with $|S|,|T|\le k$, there are $S'\subseteq S$ and $T'\subseteq T$ with $e(S',T')\ge e(S,T)/2$ such that every vertex of $S'$ has more than $e(S,T)/(4|S|)$ neighbors in $T'$ and every vertex of $T'$ has more than $e(S,T)/(4|T|)$ neighbors in $S'$.

From every starting vertex $v\in S'\mathbin{\dot\cup}T'$, the vectors of Andersen's algorithm with parameter $k$ satisfy $\|x_tA\|_2\ge(d(S,T)/16)\|x_t\|_2$ for some $t<t_0$. If $d(S,T)<16$, then also $\|x_0A\|_2\ge\|x_0\|_2$.
\end{lemma}

For $d(S,T)\ge16$, we apply Andersen's proof of Theorem~1 with $\lambda=d(S,T)/16$. For the last claim, every $v\in S'\mathbin{\dot\cup}T'$ has a neighbor, so $\|x_0A\|_2\ge1=\|x_0\|_2$.

\begin{lemma}[From growth to dense level sets {[\citealp[Lemma~1]{Andersen2010}, weakened form]}]
\label{lem:andersen-extraction}
For a nonzero vector $x\ge0$ supported on one side and $\lambda\ge1$ with $\|xA\|_2\ge\lambda\|x\|_2$, there are nonempty dyadic level sets $X$ of $x/\|x\|_2$ and $Y$ of $xA/\|xA\|_2$ with $d(X,Y)$ at least
\begin{equation}
\label{eq:andersen-extraction}
\frac{\lambda}{72\,(1+\log_2(\Delta/\lambda))}
\enspace.
\end{equation}
\end{lemma}

In Andersen's Lemma~1, the density is at least $\lambda/(72+16\log_2(\Delta/\lambda))$ for some level-set indices $i,j$ satisfying $|i-j|\le4+\log_2\Delta$. Since $\lambda\le\Delta$, \Cref{eq:andersen-extraction} is weaker, and we state it without the restriction on the level-set indices.

\subsection{Neighbor-query implementation of Andersen's algorithm}
\label{app:andersen-queries}

As in Andersen's proof of Theorem~2 \citep{Andersen2010}, the iteration in \Cref{app:andersen-lemmas} has $|\supp(x_t)|=O(k^2 4^{-t})$. At step $t$, we read the adjacency list of each vertex $u\in\supp(x_t)$ to its end, using $\deg(u)+1\le\Delta+1$ neighbor queries. Using the cached lists, we compute $x_tA$ and $e(X,Y)$ for all level sets $X$ of $x_t/\|x_t\|_2$ and $Y$ of $x_tA/\|x_tA\|_2$ without further queries. The total number of neighbor queries is therefore at most
\begin{equation*}
(\Delta+1)\sum_{t=0}^{t_0-1}|\supp(x_t)|
=O(\Delta k^2)
\enspace,
\end{equation*}
since $\Delta\ge1$ for a target of positive density. If $\Delta$ is unknown, we omit the restriction $|i-j|\le4+\log_2\Delta$ and consider all level sets $X,Y$ with $e(X,Y)>0$. No additional queries are needed, and the returned density is no smaller.

\subsection{Logarithmic accumulation for \texorpdfstring{weak-$\ell_2$}{weak-l2}}
\label{app:log-factor}

Andersen's vectors are normalized in the Euclidean norm. For $x\ge0$ supported on $L$ with $\|x\|_2\le1$, by Cauchy--Schwarz, $I_x(i)\le\sqrt i$ for every $i$, that is, $\|x\|_{2,\infty}\le1$ in the notation of \Cref{def:prefix}. We bound the prefix sums of $y\coloneq xA$ as in \Cref{sec:prefix-sweeps}, for a prefix $P_m(y)$ with $d(P_i(x),P_m(y))\le\sigma$ for every prefix $P_i(x)$ of $x$.

We also assume that every vertex in $\supp(x)$ has degree below $\sigma^2$, as in $\Search_p$ after the adjacency lists of $\supp(x)$ have been read and no star has been returned (\Cref{lem:step-outputs}(a)). By the density assumption, and since the vertices of $P_i(x)$ have degree below $\sigma^2$ and those of $P_m(y)$ degree at most $\Delta$, $e(P_i(x),P_m(y))\le g(i)$ for $g(i)\coloneq\min\{\sigma^2i,\,\sigma\sqrt{im},\,\Delta m\}$.
We write $r_0\coloneq m/\sigma^2$ and $r_1\coloneq\Delta^2m/\sigma^2$. Since $y\ne0$, $\Delta\ge1$ and $r_1/r_0=\Delta^2\ge1$. The degree bound $\sigma^2i$ is the smallest of the three for $i\le r_0$, the density bound $\sigma\sqrt{im}$ for $r_0\le i\le r_1$, and the maximum-degree bound $\Delta m$ for $i\ge r_1$. By summation by parts, as in the proof of \Cref{prop:interpolation},
\begin{equation*}
I_y(m)\le\sum_{i=1}^{\ell}x_{(i)}\bigl(g(i)-g(i-1)\bigr)
\enspace.
\end{equation*}
We split this sum at $r_0$ and $r_1$ and bound the three parts separately. The terms with $i\le r_0$ have $g(i)-g(i-1)=\sigma^2$, so their total contribution is at most $\sigma^2I_x(r_0)\le\sigma^2\sqrt{r_0}=\sigma\sqrt m$.
The terms with $i\ge r_1+1$ are zero, since $g(i)=g(i-1)=\Delta m$. If $i-1<r_0<i$ or $i-1<r_1<i$, the corresponding term is at most $x_{(i)}g(i)\le\sigma\sqrt m$. There are at most two such indices. For the remaining indices, $r_0\le i-1<i\le r_1$, so $g(i)-g(i-1)=\sigma\sqrt m(\sqrt i-\sqrt{i-1})\le\sigma\sqrt m/\sqrt i$. Since $x_{(i)}\le I_x(i)/i\le1/\sqrt i$, the terms in a dyadic range $2^j\le i<2^{j+1}$ therefore add up to
\begin{equation*}
\sum_{\substack{2^j\le i<2^{j+1}\\r_0\le i-1<i\le r_1}}x_{(i)}\bigl(g(i)-g(i-1)\bigr)
=O(\sigma\sqrt m)\sum_{2^j\le i<2^{j+1}}\frac1i
=O(\sigma\sqrt m)
\enspace.
\end{equation*}
Since $r_1/r_0=\Delta^2$, at most $2\log_2\Delta+O(1)$ dyadic ranges contain these remaining indices. Altogether,
\begin{equation}
\label{eq:ov-three}
I_y(m)=O\bigl(\sigma\sqrt m\,(1+\log_2\Delta)\bigr)
\enspace.
\end{equation}

Using only $I_x(i)\le\sqrt i$, we cannot improve the logarithmic dependence in the preceding upper bound on $\sum_ix_{(i)}(g(i)-g(i-1))$. The vector $x$ with $x_{(i)}\coloneq\sqrt i-\sqrt{i-1}$ for $i\le\ell$, where $\ell>r_1$ and $m\ge\sigma^2$, satisfies $I_x(i)=\sqrt i$, and the contribution to this sum from each dyadic range between $r_0$ and $r_1$ is $\Theta(\sigma\sqrt m)$. This vector has $\|x\|_2>1$. For $\|x\|_2\le1$, however, the sum is $O(\sigma\sqrt m\sqrt{1+\log_2\Delta})$: by Cauchy--Schwarz, it is at most $\|x\|_2\bigl(\sum_i(g(i)-g(i-1))^2\bigr)^{1/2}$, and $(g(i)-g(i-1))^2=\bigl(\int_{i-1}^ig'(s)\,ds\bigr)^2\le\int_{i-1}^ig'(s)^2\,ds$, so that
\begin{equation*}
\sum_i\bigl(g(i)-g(i-1)\bigr)^2
\le\int_0^\infty g'(s)^2\,ds
=\sigma^4r_0+\frac{\sigma^2m}4\ln\frac{r_1}{r_0}
=\sigma^2m\Bigl(1+\frac12\ln\Delta\Bigr)
\enspace.
\end{equation*}
To bound the prefix sums of $yA$ by Cauchy--Schwarz, we need an upper bound on $\|y\|_2$. We have bounded only $\|y\|_{2,\infty}$. For $p>2$, by contrast, we can bound $\|y\|_{p,\infty}$ in terms of $\|x\|_{p,\infty}$ (\Cref{prop:interpolation}).

The logarithm in Andersen's Lemma~1 also comes from a sum over logarithmically many scales: in its proof, he sums over all level-set indices $i,j$ with $|i-j|\le4+\log_2(\Delta/\lambda)$.

For $p>2$, in contrast, $x_{(i)}\le\|x\|_{p,\infty}i^{-1/p}$. The sum over the dyadic range $2^j\le i<2^{j+1}$ is then $O(\sigma\sqrt m\,\|x\|_{p,\infty}2^{j(1/2-1/p)})$. The sum over all ranges is therefore $O_p(1)$ times the largest of these bounds, and there is no logarithmic factor in \Cref{prop:interpolation}.

\subsection{A tight example for norm-growth extraction}
\label{app:extraction-example}

For nonnegative matrices, the logarithmic factor in \Cref{eq:andersen-extraction} is necessary \citep[see][]{BiluLinial2006}. For a nonnegative matrix $B$ and nonempty sets $S,T$ of row and column indices, we write
\begin{equation*}
d_B(S,T)\coloneq\frac{\sum_{u\in S}\sum_{w\in T}B_{uw}}{\sqrt{|S||T|}}
\enspace.
\end{equation*}
In \Cref{con:multiscale}, we define a symmetric matrix $\bar A$ and a vector $\xi_L$ with $\|\xi_L\bar A\|_2=\lambda\|\xi_L\|_2$ such that $d_{\bar A}(S,T)=O(\lambda/\log_2\Delta)$ for all nonempty $S\subseteq L$ and $T\subseteq R$, where $\Delta$ is the largest row sum.

\begin{construction}
\label{con:multiscale}
For an integer $\nu\ge1$ and $\zeta\coloneq2^\nu$, the sets $L$ and $R$ are partitioned into blocks $L_i,R_i$ of size $4^i$ for $i=0,\ldots,\nu-1$. The symmetric matrix $\bar A$, indexed by $L\mathbin{\dot\cup}R$, has entry $\zeta2^{-i-j}$ between $L_i$ and $R_j$ and zero entries on $L\times L$ and $R\times R$. The vector $\xi_L$ equals $2^{-i}$ on $L_i$ and zero on $R$, the vector $\xi_R$ is defined symmetrically, and $\lambda\coloneq\|\xi_L\bar A\|_2/\|\xi_L\|_2$.
\end{construction}

\begin{proposition}
\label{cor:tightness}
The matrix of \Cref{con:multiscale} has spectral radius $\nu \zeta$ and largest row sum $\Delta=\zeta(\zeta-1)$. Moreover,
\begin{equation*}
\lambda=\nu \zeta,
\qquad
\zeta\le\max_{\substack{\varnothing\ne S\subseteq L\\\varnothing\ne T\subseteq R}}d_{\bar A}(S,T)\le9\zeta
\enspace.
\end{equation*}
The nonempty dyadic level sets of $\xi_L/\|\xi_L\|_2$ and $\xi_L\bar A/\|\xi_L\bar A\|_2$ are $L_0,\ldots,L_{\nu-1}$ and $R_0,\ldots,R_{\nu-1}$, respectively. Every block density $d_{\bar A}(L_i,R_j)$ equals $\zeta$. The ratio of $\lambda$ to the maximum of $d_{\bar A}(S,T)$ over nonempty $S\subseteq L$ and $T\subseteq R$ is $\Theta(\nu)=\Theta(\log_2\Delta)$. The right-hand side of \Cref{eq:andersen-extraction}, evaluated at $\lambda=\nu\zeta$, is $\Theta(\zeta)$.
\end{proposition}

\begin{proof}
For $u\in R_j$ with $0\le j<\nu$, the contribution from each left block $L_i$ to $(\xi_L\bar A)(u)$ is $4^i\cdot2^{-i}\cdot\zeta2^{-i-j}=\zeta2^{-j}$, so
\begin{equation*}
(\xi_L\bar A)(u)=\nu\zeta2^{-j}=\nu \zeta\,\xi_R(u)
\enspace.
\end{equation*}
By symmetry, $\xi_R\bar A=\nu\zeta\,\xi_L$, so $\xi_L+\xi_R$ is a positive eigenvector with eigenvalue $\nu \zeta$, which is therefore the spectral radius. Since $\|\xi_L\|_2=\|\xi_R\|_2=\sqrt \nu$, $\lambda=\nu\zeta$. By direct summation, the largest row sum is $\Delta=\zeta(\zeta-1)$, and $d_{\bar A}(L_i,R_j)=\zeta$ for all $i,j$. Consecutive block values differ by a factor of two, so the nonempty dyadic level sets are exactly the blocks.

For nonempty $S\subseteq L$ and $T\subseteq R$, we write $s_i\coloneq|S\cap L_i|$ and take the integer $q\ge0$ with $\sqrt{|S|}\le2^q<2\sqrt{|S|}$. Since $s_i\le4^i$,
\begin{equation*}
\sum_{i=0}^{\nu-1}s_i2^{-i}\le\sum_{i=0}^{q-1}2^i+2^{-q}|S|
\le2^q+|S|2^{-q}\le3\sqrt{|S|}
\enspace.
\end{equation*}
We apply the same bound to $t_j\coloneq|T\cap R_j|$ and obtain
\begin{equation*}
\sum_{u\in S}\sum_{w\in T}\bar A_{uw}
=\zeta\Bigl(\sum_{i=0}^{\nu-1}s_i2^{-i}\Bigr)\Bigl(\sum_{j=0}^{\nu-1}t_j2^{-j}\Bigr)
\le9\zeta\sqrt{|S||T|}
\enspace.
\end{equation*}
Finally, since $\Delta=\zeta(\zeta-1)$ and $\lambda=\nu\zeta$, both $\log_2\Delta$ and $1+\log_2(\Delta/\lambda)$ are $\Theta(\nu)$, so the right-hand side of \Cref{eq:andersen-extraction} is $\Theta(\zeta)$.
\end{proof}

\section{Lower bounds}
\label{app:limits}

Since sets $X,Y$ with $d(X,Y)=\Omega(k)$ have $\Omega(k^2)$ edges between them, an output of $o(k^2)$ vertices per side must have many more edges than vertices, as the target does, and the algorithm has to find these edges. The query lower bound is proved by hiding a complete bipartite graph at uniformly chosen adjacency-list positions (\Cref{app:queries}), and the computational lower bounds by reducing Densest $k$-Subgraph to the local problem (\Cref{app:computational}).

\begin{definition}
\label{def:algorithmic-setting}
An algorithm $\mathsf A$ receives a starting vertex $v$, an integer $k\ge1$, and neighbor-query access to a simple bipartite graph $G$, as specified in the graph-access model of \Cref{app:computation}. The algorithm has \defemph{output bound} $\phi$ if, for every input, it returns either $\bot$ or nonempty sets $X\subseteq L$, $Y\subseteq R$ with $|X|,|Y|\le\phi(k)$. Unless we explicitly allow randomness, algorithms are deterministic; for randomized algorithms, all output, query, and running-time bounds hold for every choice of the internal randomness, and success probabilities are over that randomness. Unless we impose a running-time bound, computation between queries is unrestricted.
\end{definition}

\subsection{The query lower bound}
\label{app:queries}

\begin{theorem}[Queries sublinear in $\Delta$]
\label{thm:queries}
For a constant $0<\rho\le1$ and $\phi(k)=o(k^2)$, if a possibly randomized algorithm $\mathsf A$ as in \Cref{def:algorithmic-setting} has output bound $\phi$ and makes at most $q(k,\Delta)$ queries on every input, where $q(k,\Delta)=o(\Delta)$ as $\Delta\to\infty$ for every fixed $k$, then, for every sufficiently large $k$, there is a simple bipartite graph $G$ containing a complete bipartite subgraph with sides $S\subseteq L$ and $T\subseteq R$ of $k$ vertices each, such that from every $v\in S\mathbin{\dot\cup}T$ the output of $\mathsf A$ is $\bot$ or has density below $\rho k$ with probability at least $2/3$.
\end{theorem}

\begin{proof}[Proof of \Cref{thm:queries}]
For the extension to randomized algorithms, we use a transcript-coupling approach inspired by \citet[Section~5.2]{PengWangYangYang2026}.
Since $\phi(k)=o(k^2)$, we can fix any sufficiently large $k$ with $\sqrt{2\phi(k)}<\rho k$. Choose integers $N>k$ and then $\Delta\ge k$ so that
\begin{equation*}
\frac{2k^2q(k,\Delta)}{\Delta}+\frac{4\phi(k)k^2}{N-1}<\frac13,
\end{equation*}
which is possible because $q(k,\Delta)=o(\Delta)$ for fixed $k$.
Let $G_0$ consist of $N$ disjoint stars centered on each side, each with $\Delta$ leaves, with fixed vertex identifiers and adjacency-list orders. Choose $S$ and $T$ independently and uniformly among the $k$-subsets of the left and right centers, respectively. For each chosen center, independently select $k$ of the $\Delta$ positions in its adjacency list uniformly, and replace the entries at those positions by the $k$ target centers on the opposite side in identifier order. Isolate the corresponding leaves, keeping all vertex identifiers. The resulting graph $G$ is simple and bipartite, every center still has degree $\Delta$, and $S,T$ span a complete bipartite subgraph of density $k$.

Fix a center $v$ and condition on $v\in S\mathbin{\dot\cup}T$. Fix the internal randomness of $\mathsf A$ independently of the construction, and compare its runs from $v$ on $G_0$ and $G$ using this same randomness. The run on $G_0$ is independent of the chosen target and positions. Each query in that run concerns at most one original star edge, either through its position in a center's list or through a query at its leaf. Conditional on $v\in S\mathbin{\dot\cup}T$, each such edge is replaced with probability at most $k/\Delta$. Until a differing answer is received, the two computations coincide. By a union bound over the at most $q(k,\Delta)$ queries on $G_0$, the probability of any differing answer is at most $kq(k,\Delta)/\Delta$.

The output on $G_0$ contains at most $2\phi(k)$ centers. Each center other than $v$ belongs to $S\mathbin{\dot\cup}T$ with conditional probability at most $k/(N-1)$, so the probability that the output contains another target center is at most $2\phi(k)k/(N-1)$. If neither event occurs, the output on $G$ is the same as on $G_0$ and contains no edge of $S\times T$. If it is a pair $(X,Y)$ rather than $\bot$, all its edges belong to $G_0$. Writing $e_0\coloneq e_{G_0}(X,Y)$, we have $e_0\le|X|+|Y|\le2\phi(k)$ since $G_0$ is a forest, and $e_0\le|X||Y|$. Thus
\begin{equation*}
d_G(X,Y)\le\frac{e_0}{\sqrt{|X||Y|}}\le\sqrt{e_0}\le\sqrt{2\phi(k)}<\rho k.
\end{equation*}

Let $p_v(G)$ be the probability, over the internal randomness of $\mathsf A$, of returning sets of density at least $\rho k$ from $v$ on $G$. Averaging also over that randomness in the preceding bounds, we obtain
\begin{equation*}
\mathbb E\!\left[p_v(G)\mid v\in S\mathbin{\dot\cup}T\right]
\le\frac{kq(k,\Delta)}{\Delta}+\frac{2\phi(k)k}{N-1},
\end{equation*}
where the expectation is over the random construction. There are $2N$ centers, each belonging to $S\mathbin{\dot\cup}T$ with probability $k/N$, so
\begin{equation*}
\mathbb E\!\left[\sum_{v\in S\mathbin{\dot\cup}T}p_v(G)\right]
\le\frac{2k^2q(k,\Delta)}{\Delta}+\frac{4\phi(k)k^2}{N-1}<\frac13.
\end{equation*}
Hence there is one fixed graph and target for which this sum is below $1/3$. Every summand is nonnegative, so $p_v(G)<1/3$ for every $v\in S\mathbin{\dot\cup}T$, as required.
\end{proof}

\begin{corollary}[Approximation and output size with restricted queries]
\label{cor:query-tradeoff}
For $\alpha(k)\ge1$ with $\alpha(k)=o(k)$ and $\phi(k)=o(k^2/\alpha(k)^2)$, if a possibly randomized algorithm as in \Cref{def:algorithmic-setting} has output bound $\phi$ and makes at most $q(k,\Delta)$ queries, where $q(k,\Delta)=o(\Delta)$ as $\Delta\to\infty$ for every fixed $k$, then, for every sufficiently large $k$, there is a simple bipartite graph containing a complete bipartite subgraph with sides $S,T$ of $k$ vertices each, such that, from every $v\in S\mathbin{\dot\cup}T$, the algorithm returns $\bot$ or sets of density below $d(S,T)/\alpha(k)$ with probability at least $2/3$.
\end{corollary}

\begin{proof}
We apply the construction in the proof of \Cref{thm:queries} with $\rho=1/\alpha(k)$, since its condition $\sqrt{2\phi(k)}<\rho k$ holds for all sufficiently large $k$:
\begin{equation*}
\frac{\alpha(k)}{k}\sqrt{2\phi(k)}
=\sqrt{\frac{2\alpha(k)^2\phi(k)}{k^2}}
<1
\enspace,
\end{equation*}
as the expression tends to zero. For each such $k$, choose $N$ and $\Delta$ as in that proof. From every target vertex of the resulting fixed graph, the algorithm returns $\bot$ or sets of density below $k/\alpha(k)=d(S,T)/\alpha(k)$ with probability at least $2/3$.
\end{proof}

\subsection{Computational lower bounds}
\label{app:computational}

For both lower bounds under SPCH and ETH, we use the same reduction. In \Cref{app:reduction}, we construct a dense subgraph on $k$ vertices of a graph $J$ from the outputs of a local algorithm on its bipartite double cover, and in \Cref{app:spch,app:eth}, we prove the two lower bounds. Polynomial time means time polynomial in the number of vertices of the input graph, uniformly over all $k$ up to that number.

\subsubsection{Reduction from Densest \texorpdfstring{$k$}{k}-Subgraph}
\label{app:reduction}

For a simple graph $J$ on $N$ vertices, $e_J(U)$ is the number of edges induced by $U\subseteq V(J)$, and the \defemph{bipartite double cover} $G$ of $J$ has two copies $L,R$ of $V(J)$ and edges $u_Lw_R,w_Lu_R$ for every $uw\in E(J)$. The graph $G$ has $2N$ vertices, and each vertex of $G$ has the same degree as the corresponding vertex of $J$, which is below $N$. For $k\ge2$, the two copies $S\subseteq L$ and $T\subseteq R$ of a $k$-clique in $J$ are cohesive with $d(S,T)=k-1$, since every vertex has $k-1$ neighbors in the opposite set.

\begin{lemma}[Extracting a dense $k$-vertex subgraph]
\label{lem:extraction}
For a simple graph $J$ on at least $k\ge2$ vertices and its bipartite double cover $G$, $0<\rho\le1$, $B\ge1$, and nonempty sets $X\subseteq L$ and $Y\subseteq R$ with $|X|,|Y|\le Bk$ and $d(X,Y)\ge\rho(k-1)$, a set $U$ of $k$ vertices of $J$ can be found in polynomial time with
\begin{equation*}
e_J(U)\ge\frac{\rho^2}{16B}\binom k2
\enspace.
\end{equation*}
\end{lemma}

\begin{proof}
For $h\coloneq\lfloor k/2\rfloor$, we take $X'$ to be the $\min\{|X|,h\}$ vertices of $X$ with the most neighbors in $Y$, and $Y'$ the $\min\{|Y|,h\}$ vertices of $Y$ with the most neighbors in $X'$.

Since $|X'|+|Y'|\le2h\le k$, we form $U$ from the corresponding vertices of $J$, adding arbitrary vertices until $|U|=k$. Each edge of $J$ has at most two copies in the double cover, so $e_J(U)\ge e(X',Y')/2$.

The vertices of $X'$ have average degree into $Y$ at least $e(X,Y)/|X|$, so $e(X',Y)\ge(|X'|/|X|)e(X,Y)$, and similarly $e(X',Y')\ge(|Y'|/|Y|)e(X',Y)$. Therefore,
\begin{equation*}
e(X',Y')\ge\frac{|X'|}{|X|}\frac{|Y'|}{|Y|}e(X,Y)
=\frac{h^2e(X,Y)}{\max\{|X|,h\}\max\{|Y|,h\}}
\enspace.
\end{equation*}

By the density assumption, $e(X,Y)\ge\rho(k-1)\sqrt{|X||Y|}$ and, since $d(X,Y)\le\sqrt{e(X,Y)}$, also $e(X,Y)\ge\rho^2(k-1)^2$. If one side has at most $h$ vertices, one factor in the denominator of the averaging bound is $h$ and the other is at most $Bk$. We therefore keep at least an $h/(Bk)$ fraction of the edges, so $e(X',Y')\ge h\rho^2(k-1)^2/(Bk)$.
Otherwise, $|X|,|Y|>h$, so the denominator in the averaging bound is $|X||Y|$. Since $\sqrt{|X||Y|}\le Bk$,
\begin{equation*}
e(X',Y')\ge\frac{h^2\rho(k-1)}{\sqrt{|X||Y|}}
\ge\frac{h^2\rho(k-1)}{Bk}
\ge\frac{h\rho^2(k-1)^2}{2Bk}
\enspace,
\end{equation*}
using $h\ge(k-1)/2$ and $\rho\le1$. In either case,
\begin{equation*}
e_J(U)\ge\frac{h\rho^2(k-1)^2}{4Bk}
\ge\frac{\rho^2}{16B}\binom k2
\enspace,
\end{equation*}
using $h\ge(k-1)/2$ and $k-1\ge k/2$.
\end{proof}

\begin{lemma}
\label{lem:reduction}
For a simple graph $J$ on $N\ge k\ge2$ vertices and its bipartite double cover $G$, $0<\rho\le1$, $B\ge1$, and an algorithm $\mathsf A$ that returns $\bot$ or nonempty sets $X\subseteq L$ and $Y\subseteq R$ with $|X|,|Y|\le Bk$ from every $v\in V(G)$, we can find, with $2N$ calls to $\mathsf A$ and polynomial additional time, a set $U$ of $k$ vertices of $J$ such that, if $\mathsf A(G,v,k)$ returns sets $X,Y$ with $d(X,Y)\ge\rho(k-1)$ for some $v$, then
\begin{equation*}
e_J(U)\ge\frac{\rho^2}{16B}\binom k2
\enspace.
\end{equation*}
\end{lemma}

\begin{proof}
We run $\mathsf A$ from every vertex of $G$, apply the construction in \Cref{lem:extraction} to each output $X,Y$, and keep a set inducing the most edges. If every call returns $\bot$, we take any $k$ vertices. If $\mathsf A$ returns sets $X,Y$ with $d(X,Y)\ge\rho(k-1)$, then by \Cref{lem:extraction} the corresponding set of $k$ vertices satisfies the claimed bound, as does the set we keep. We construct the double cover and select the sets of $k$ vertices in polynomial time.
\end{proof}

\subsubsection{The limitation under SPCH}
\label{app:spch}

\begin{definition}[Strongish Planted Clique Hypothesis {[\citealp[Hypothesis~1]{ManurangsiRubinsteinSchramm2021}]}]
\label{def:spch}
The \defemph{Strongish Planted Clique Hypothesis} (SPCH) is the assumption that, for some $0<\theta<1/2$, no randomized algorithm running in time $n^{o(\log_2n)}$ distinguishes the random graph $\mathcal G(n,\frac12)$ on $n$ vertices, in which each edge is present independently with probability $1/2$, from the same graph $\mathcal G(n,\frac12,\lceil n^\theta\rceil)$ with a clique added on $\lceil n^\theta\rceil$ uniformly chosen vertices, with success probability at least $2/3$ under each distribution. We call these the \textsc{no} and \textsc{yes} distributions, respectively. The probability is over the input graph and the algorithm's randomness.
\end{definition}

We use the following hardness theorem for Densest $k$-Subgraph, with the promise that the input contains a $k$-clique.

\begin{theorem}[{[\citealp[Theorem~2]{ManurangsiRubinsteinSchramm2021}]}]
\label{thm:mrs}
Assuming SPCH, for every $\beta$ with $\beta(k)=o(k)$ and every computable $f$, no algorithm, given an $N$-vertex graph $J$ and $k$ such that $J$ contains a $k$-clique, returns in time $f(k)N^{O(1)}$ a set $U$ of $k$ vertices with
\begin{equation*}
e_J(U)\ge\frac1{\beta(k)}\binom k2
\enspace.
\end{equation*}
\end{theorem}

\begin{theorem}[Limitations on the output size under SPCH]
\label{thm:quadratic}
Assuming SPCH, for a constant $0<\rho\le1$ and $\phi(k)=o(k^2)$, if a possibly randomized algorithm $\mathsf A$ as in \Cref{def:algorithmic-setting} has output bound $\phi$ and runs in time $f(k)|V|^{O(1)}$ on every input $(G,v,k)$, for a computable function $f$ of $k$ alone, then, for some $k\ge2$, there is a simple bipartite graph $G$ with cohesive sets $S\subseteq L$ and $T\subseteq R$ of $k$ vertices each, such that from every $v\in S\mathbin{\dot\cup}T$ the output is $\bot$ or has density below $\rho\,d(S,T)$ with probability greater than $1/3$.
\end{theorem}

\begin{proof}[Proof of \Cref{thm:quadratic}]
Suppose, for a contradiction, that for every $k\ge2$ and all cohesive target sets $S\subseteq L$ and $T\subseteq R$ of $k$ vertices each, the algorithm $\mathsf A$ returns sets $X,Y$ with $d(X,Y)\ge\rho\,d(S,T)$ from some $v\in S\mathbin{\dot\cup}T$ with probability at least $2/3$. For $J$ with a $k$-clique, the two copies $S,T$ of the clique in the double cover are cohesive with $d(S,T)=k-1$. We run $\mathsf A$ from every vertex of $G$ and form a set of $k$ vertices of $J$ from each returned pair $X,Y$ as in the proof of \Cref{lem:extraction}. Among these sets, we keep one with the most edges between its vertices. By \Cref{lem:reduction} with $B=1+\phi(k)/k$, we can find a set $U$ of $k$ vertices inducing at least $\binom k2/\beta(k)$ edges with probability at least $2/3$, where
\begin{equation*}
\beta(k)\coloneq\frac{16(k+\phi(k))}{\rho^2k}=o(k)
\enspace.
\end{equation*}
We can compute $e_J(U)$ for each returned set, so we repeat the procedure in \Cref{lem:reduction} independently a constant number of times (concretely, five) and keep a set inducing the most edges. The subgraph induced by this set contains at least $\binom k2/\beta(k)$ edges with probability at least $1-3^{-5}>0.99$. The running time is $f'(k)N^{O(1)}$ for a computable $f'$.

In the reduction from planted clique in \citet[Section~4.1]{ManurangsiRubinsteinSchramm2021}, completeness and soundness each hold with probability at least $0.9$. We use the procedure in \Cref{lem:reduction}, with the five independent repetitions, as the Densest $k$-Subgraph algorithm in that reduction. We then distinguish the \textsc{yes} and \textsc{no} distributions of SPCH (\Cref{def:spch}) with success probability at least $0.9\cdot0.99>2/3$ under each distribution, contradicting SPCH.
\end{proof}

\begin{remark*}
We do not require $\beta$ in \Cref{thm:mrs} to be nondecreasing. In \citet[Section~4.1]{ManurangsiRubinsteinSchramm2021}, the ratio is a function $g(k)\le k$ growing with $k$, so we replace $\beta$ with the nondecreasing $\tilde\beta(k)\coloneq\max_{1\le j\le k}\beta(j)$. For every $\delta>0$, with $C_\delta\ge0$ bounding $\beta(j)$ for the finitely many $j$ with $\beta(j)>\delta j$, $\tilde\beta(k)\le\max\{C_\delta,\delta k\}$, so $\tilde\beta=o(k)$. In particular, $\tilde\beta(k)\le k$ for all sufficiently large $k$, and exhaustive search handles the remaining values of $k$.
\end{remark*}

\begin{corollary}[Approximation and output size under SPCH]
\label{cor:spch-tradeoff}
Assuming SPCH, for $\alpha(k)\ge1$ with $\alpha(k)=o(\sqrt{k})$ and $\phi(k)=o(k^2/\alpha(k)^2)$, no algorithm as in \Cref{def:algorithmic-setting} with output bound $\phi$ and running time $f(k)|V|^{O(1)}$ for a computable $f$ guarantees, for every $k\ge2$ and every cohesive target $S,T$ of $k$ vertices per side, an output of density at least $d(S,T)/\alpha(k)$ from some starting vertex in $S\mathbin{\dot\cup}T$.
\end{corollary}

\begin{proof}
For such an algorithm, by \Cref{lem:reduction} with $\rho=1/\alpha(k)$ and $B=1+\phi(k)/k$, we obtain an approximation ratio $\beta(k)\coloneq16\alpha(k)^2(k+\phi(k))/k$ for Densest $k$-Subgraph, with
\begin{equation*}
\frac{\beta(k)}{k}
=16\left(\frac{\alpha(k)^2}{k}+\frac{\alpha(k)^2\phi(k)}{k^2}\right)
\longrightarrow0
\enspace.
\end{equation*}
The running time is $f'(k)N^{O(1)}$ for a computable $f'$, contradicting \Cref{thm:mrs}.
Following the proof of \Cref{thm:quadratic} with $\rho=1/\alpha(k)$, we also rule out algorithms that use randomness, always respect the bounds on output size and running time, and succeed with probability at least $2/3$.
\end{proof}

\paragraph{Quadratic output on the reduction instances.}
When $J$ contains a $k$-clique, we can achieve a quadratic output bound on its bipartite double cover in time $k^{O(k)}N^{O(1)}$. If a vertex of $J$ has at least $k^2$ neighbors, one of its copies and $k^2$ of that copy's neighbors form a star of density $k$. Otherwise, every neighborhood in $J$ has fewer than $k^2$ vertices, and we find the promised $k$-clique by testing all $(k-1)$-subsets of each neighborhood together with its center. The two copies of the clique have density $k-1$. Thus, on these instances, subquadratic output is essential to the lower bound even with full graph access.

\subsubsection{The limitation under ETH}
\label{app:eth}

\begin{definition}[Exponential Time Hypothesis {[\citealp{ImpagliazzoPaturi2001}]}]
\label{def:eth}
The \defemph{Exponential Time Hypothesis} (ETH) is the assumption that 3-SAT for formulas with $m$ clauses cannot be decided in time $2^{o(m)}$. By the sparsification lemma of \citet{ImpagliazzoPaturiZane2001}, this is equivalent to the formulation in terms of variables.
\end{definition}

We use the following hardness theorem, in which the graph contains a $k$-clique in the \textsc{yes} case.

\begin{theorem}[{[\citealp[Theorem~1]{Manurangsi2017}]}]
\label{thm:manurangsi}
Assuming ETH, there is a constant $c_0>0$ such that no polynomial-time algorithm, given a graph $J$ on $N\ge4$ vertices and an integer $2\le k\le N$, distinguishes the case that $J$ contains a $k$-clique from the case that every $k$-vertex subgraph of $J$ contains at most $\binom k2/\chi(N)$ edges, where
\begin{equation*}
\chi(N)\coloneq N^{1/(\log_2\log_2N)^{c_0}}
\enspace.
\end{equation*}
\end{theorem}

\begin{theorem}[Limitations on the output size under ETH]
\label{thm:limits}
Assuming ETH, for constants $0<\rho\le1$ and $c\ge1$, if a deterministic polynomial-time algorithm $\mathsf A$ returns $\bot$ or nonempty sets $X\subseteq L$ and $Y\subseteq R$ with $|X|,|Y|\le c\,k\log_2^c(2+k+\Delta)$ on every input $(G,v,k)$, then, for some $k\ge2$, there is a simple bipartite graph $G$ with cohesive sets $S\subseteq L$ and $T\subseteq R$ of $k$ vertices each, such that from every $v\in S\mathbin{\dot\cup}T$ the output is $\bot$ or has density below $\rho\,d(S,T)$. The same conclusion holds for running time polynomial in $k$, $\Delta$ and $\log_2|V|$.
\end{theorem}

\begin{proof}
Suppose, for a contradiction, that for all cohesive target sets $S\subseteq L$ and $T\subseteq R$, the algorithm $\mathsf A$ returns sets $X,Y$ with $d(X,Y)\ge\rho\,d(S,T)$ from some $v\in S\mathbin{\dot\cup}T$. For an $N$-vertex graph $J$ and $k\le N$, the double cover has $\Delta<N$, so the output bound is at most $Bk$ for $B\coloneq c\lceil\log_2(4N)\rceil^c$, and we apply the procedure of \Cref{lem:reduction}. If $J$ contains a $k$-clique, its two copies $S,T$ are cohesive with $d(S,T)=k-1$, so $\mathsf A$ returns sets $X,Y$ with $d(X,Y)\ge\rho(k-1)$ from some vertex, and the set $U$ of \Cref{lem:reduction} satisfies
\begin{equation*}
e_J(U)\ge\frac{\rho^2}{16B}\binom k2
>\frac1{\chi(N)}\binom k2
\enspace
\end{equation*}
for all sufficiently large $N$, since $\chi(N)$ grows faster than every power of $\log_2N$. In the other case of the promise, every set of $k$ vertices, including $U$, induces at most $\binom k2/\chi(N)$ edges. By comparing $e_J(U)$ with $\binom k2/\chi(N)$, we can therefore distinguish the two cases of \Cref{thm:manurangsi} in polynomial time, with exhaustive search for bounded $N$. Finally, a running time polynomial in $k$, $\Delta$ and $\log_2|V|$ is polynomial in $|V|$ when $k\le|V|$ and $\Delta<|V|$.
\end{proof}

\paragraph{Why the limitation under ETH is weaker.}

We cannot obtain the quadratic bound of \Cref{thm:quadratic} from the hard instances used under ETH, because they have $N=k^{1+o(1)}$ vertices. \citet[Section~3.2]{Manurangsi2017} constructs graphs whose vertices are partial assignments to $\ell$ of the $n$ variables of a 3-SAT formula. Each graph has $N=2^\ell\binom n\ell$ vertices. If the formula is satisfiable, the restrictions of a satisfying assignment to all $\ell$-element sets of variables form a clique of size $k=\binom n\ell$. With the parameter choices in that proof, $\ell=o(n)$. Hence
\begin{equation*}
\frac{\log_2N}{\log_2k}
=1+\frac{\ell}{\log_2k}
\le1+\frac1{\log_2(n/\ell)}
=1+o(1)
\enspace.
\end{equation*}
For fixed $0<\delta<1$ and large $k$, the entire double cover has fewer than $k^{2-\delta}$ vertices per side, and nonempty subsets $X\subseteq L$ and $Y\subseteq R$ maximizing $d(X,Y)$ can then be found in polynomial time \citep[Section~4, Theorem~3]{Charikar2000}, meeting the output bound $k^{2-\delta}$ with approximation ratio one.

\end{document}